\documentclass[11pt,oneside]{article}
\usepackage{amssymb, mathrsfs, amsthm} 
\usepackage[centertags]{amsmath}
\usepackage[utf8]{inputenc}
\usepackage[T1]{fontenc}
\usepackage{endnotes}
\usepackage{graphicx}
\usepackage[margin=1.25in]{geometry}
\usepackage{bbm}
\usepackage[noblocks,affil-sl]{authblk}             
\usepackage{floatrow}
\makeatletter

\providecommand{\tabularnewline}{\\}

\@ifundefined{date}{}{\date{}}

\newtheorem{definition}{Definition}
\newtheorem{remark}[definition]{Remark}

\newtheorem{theorem}[definition]{Theorem}
\newtheorem{proposition}[definition]{Proposition}

\newcommand\de{\mathrm{d}}

\makeatother

\begin{document}
\title{A Quantization Approach to the Counterparty Credit Exposure Estimation}

\author[1]{M. Bonollo \thanks{michele.bonollo@imtlucca.it}} 
\author[2]{L. Di Persio \thanks{luca.dipersio@univr.it}} 
\author[3]{I. Oliva \thanks{immacolata.oliva@univr.it}} 
\author[4]{A. Semmoloni \thanks{andrea.semmoloni@bancaprofilo.it}}

\affil[1]{Iason Ltd and IMT Lucca}
\affil[2,3]{Dept. of Computer Sciences, University of Verona, Italy}
\affil[4]{Banca Profilo}

\maketitle 
\begin{abstract}
During recent years the counterparty risk subject has received a growing attention
because of the so called {\it Basel Accord}. 
In particular the {\it Basel III Accord} asks the banks to fulfill finer conditions concerning
counterparty credit exposures arising from banks' derivatives, securities financing transactions, 
default and  downgrade risks characterizing the  Over The Counter (OTC) derivatives market, etc.
Consequently the development of effective and more accurate measures of risk have been pushed, particularly focusing on 
 the estimate of the future fair value of derivatives with respect to prescribed time horizon and fixed grid of time buckets .
Standard methods used to treat the latter scenario are mainly based on ad hoc implementations of the classic Monte Carlo (MC) approach,
which is characterized by a high computational time, strongly dependent on the number of considered assets.
This is why many financial players moved to more enhanced Technologies, e.g.,
\emph{grid computing} and {\it Graphics Processing Units} (GPUs) capabilities. 
In this paper we
show how to implement the \emph{quantization} technique,
in order to accurately estimate both pricing  and volatility values.
Our approach is tested to produce effective results for the counterparty risk evaluation, with a big improvement concerning
required time to run when compared to MC approach.




\end{abstract}

\section{Introduction and scope of the study}

\label{intro}
The financial crisis in 2007-2008, along with a consequent  increasing awareness
about the different sources of risk, has suggested to the various financial  players to give a greater attention to 
the \emph{counterparty credit risk} (CCR). 
CCR refers to
the situation when the counterparty \emph{A} has a deal, mainly of 
derivative type, such as an option or a swap, subscribed with the
counterparty \emph{B}. We suppose that, according to a valuation criteria
based on market prices,  \emph{A} observes a positive {\it fair value}, namely the so called
\textit{Mark-to-Market} (\emph{MtM}). It follows that \textit{A} has a credit
exposure with \emph{B}, hence, if \emph{B} defaults and no future
recovery rates or collateral was posted, then \emph{A} loses exactly
MtM, which is the cost for the replacement of the defaulted position. Such
type of risk is of particular interest within the  \textit{Over
The Counter} (OTC) derivatives markets, namely those markets
which are characterized by having  transactions settled directly
between the two counterparties and outside the stock exchange.

A slightly different perspective of CCR needs to be taken into account
when, in the risk management field, \textit{A} wants to assess \emph{ex-ante}
the risk belonging to the financial position underwritten with \textit{B}.
In such a case, considering the possible default  for \textit{B} as a random event both in time and in its magnitude,
 it turns out that the current MtM is a rather rough measure of the credit exposure of \textit{A}.
A better approach is given by considering the \emph{Exposure At Default} (EAD) parameter which can be seen 
as conservative expected value of the future MtM at the (random)
default time. 
EAD parameter can be seen
as conservative expected value of the future MtM at the
default time. 
An official way to estimate the EAD in various contexts
is given in the Basel framework, see \cite{bcbs0}, namely within the set
of recommendations on banking laws and regulations issued by the Basel
Committee on Banking Supervision. Such an approach is based on the
\emph{Expected Positive Exposure} (\emph{EPE}) evaluation, namely on a prudent probabilistic
time average of the future MtM. EAD follows just as a multiple, i.e.
$EAD=\alpha\cdot EPE$.

Moreover, we recall that the international accounting standards require
that in the derivatives evaluation a \emph{full fair value principle}
has to be satisfied, see, e.g., \cite{Mikes}.

If the the counterparty solvency level falls, we observe a downgrade
in its \emph{rating} and/or an increase in its \emph{spread}, therefore
the related OTC balance sheet evaluation has to embody this effect.
This implies that we have to adjust the MtM since it may decrease
not only due to the usual market parameters, e.g. underlying price,
underlying volatility, free risk rate, etc., but also because
of the credit spread volatility.

We refer to such an MtM adjustment as \textit{Credit value Adjustment}
(CVA), and the related adjusted fair value is sometimes called the
\emph{full fair value}. The adjusted MtM will be denoted by $MtM^{A}$
and we have $MtM^{A}=MtM-CVA$.

Even if the derivative has not been closed, the CVA effect can cause
a loss in the balance sheet, namely an \emph{unrealized} loss. The
Basel Committee estimates that $2/3$ of the losses in the financial
crisis years in the OTC sector were unrealized losses in the evaluation
process. The CVA (expected) loss is (or should be) absorbed by the balance sheet,
while the CVA volatility must be faced by the regulatory capital.
To this end, a new capital charge, the \emph{CVA charge}, was introduced
within the Basel III framework. We refer the reader to \cite{IFRS}
for a skillful analysis of the accounting principles and to \cite{bcbs1}
for a detailed discussion about the capital charge.

The EAD and the CVA computations pose a lot of methodological, financial
and numerical issues, as witnessed by a huge amount of literature
developed so far, see, e.g., \cite{Ces}, for a detailed review.

The present paper aims at studying the feasibility and the trade off
accuracy vs. computational effort of the \emph{quantization} approach
for the EAD-estimation (EPE), not at discussing the usefulness of
EAD/CVA measures, nor the related underlying or volatility models,
nor even at analyzing data quality and data availability. Therefore,
our main goal is the numerical CCR analysis, while we will address
the CVA issue in a future work.

In particular, we will consider a simple Black and Scholes model,
without taking into consideration collateral parameters in order to
focus the attention on the implemented numerical techniques.

The paper is organized as follows: Section \ref{epe_def} is a review
of the EPE definition given by the Basel Committee, while in Section
\ref{q_review} we give a description of the quantization approach
to the EPE with some theoretical results. Section \ref{proposal}
describes some practical cases and contains the set up of the associated
numerical experiments, finally in Section \ref{num} we report the
obtained numerical results along with their interpretation.

\section{The Basel EPE definition}\label{BaselEPESection}

\label{epe_def}

In what follows we shall give a review about the Basel Committee guidelines
concerning the estimation of the Exposure at Default, i.e. the EPE
parameter. Let us set the following notations that will be used throughout
the paper. 
\begin{itemize}
\item Given a derivative maturity time $0<T<+\infty$, we consider $K\in\mathbb{N}^{+}$
time steps $0<t_{1}<t_{2}<\cdots<t_{K}$ which constitute the so called
\emph{buckets array}, denoted by ${\bf B}^{T,K}$, where usually,
but not mandatory, $t_{K}=T$. 
\item For every $t_{k}\in{\bf B}^{T,K}$ we denote by $MtM\left(t_{k},S_{k}\right):=MtM\left(t_{k},S_{t_{k}}\right)$
the fair value (\emph{Mark-to-Market}) of a derivative at time bucket
$t_{k},$ with respect to the \emph{underlying value} $S_{k},$ at
time $t_{k}$.

For the sake of simplicity, we denote by $t_{0}=0$ the starting time
of the evaluation problem, by considering the European case.

\item For every $t_{k}\in{\bf B}^{T,K}$ we denote by $MtM\left(t_{k},S^{k}\right):=MtM\left(t_{k},S^{t_{k}}\right)$
the fair value (\emph{Mark-to-Market}) of a derivative at time bucket
$t_{k}$, with respect to the whole sample path $S^{k}:=\left\{ S_{t}:0\leq t\leq t_{k}\right\} $,
with initial time $t_{0}=0$. 
\item Taking into account previous definitions, we indicate by $\varphi=\varphi\left(T-t_{k},S_{k},\Theta\right)$
the pricing function for the given derivative, where $\Theta$ represents
the set of parameters from which such a pricing function may depends,
e.g., the free risk rate $r$ or the volatility $\sigma.$ 
\item We will use the notation $\phi^{MtM}$ to denote the Mark-to-Market
value pricing function. 
\end{itemize}
\begin{remark} We would like to underline that, in the Black-Scholes
framework, the volatility surface has to be flat, which does not occur
when real financial time series are considered. It follows that the
above-mentioned pricing function $\phi$ most likely depends on more
than the two considered parameters $r$ and $\sigma$. In particular,
usually $\Theta\in\mathbb{R}^{n}$, with $n>2$, where the extra parameters
characterize the specific geometric structure of the \textit{volatility
surface} associated to the considered contingent claim. \end{remark}

As usual in the counterparty credit risk EAD estimation, we stress
the role of the underlying, understood as the only stochastic market
parameter, while the others are deemed to be given, specifically,
we assume that they are deterministic and constant
or substituted by their deterministic forward values.

Henceforward, we shall often use the notation $\mathit{k}$ to indicate
quantities of interest evaluated at the $k-$th time bucket $\mathrm{t_{k}}.$
Besides, we give an account of the main amounts that will be used
in the following for EAD estimation, as they are defined in Basel
III, \cite{bcbs1}. 
\begin{itemize}
\item We denote the \textit{Expected Exposure} (EE) of the derivative by
\[
EE_{k}:=\frac{1}{N}\sum_{n=1}^{N}MtM\left(t_{k},S_{k,n}\right)^{+},\; N\,\in\,\mathbb{N}^{+}\;,
\]
which is nothing but the arithmetic mean of $N$ Monte Carlo simulated
MtM values, computed at the $k-$th time bucket $t_{k}$, with respect
to the underlying $\mathrm{S}.$

The\emph{ positive part} operator $\left(\cdot\right)^{+}$is effective
if we are managing a symmetric derivative, such as an interest rate
swap or a portfolio of derivatives. For a single option, it is redundant,
as the fair value of the option is always positive from the buy side
situation. We want to stress that the sell side does not imply counterparty
risk, hence it is out of context.

\item We evaluate the \emph{Expected Positive Exposure} (EPE) by 
\[
EPE:=\frac{\sum_{k=1}^{K}EE_{k}\cdot\Delta_{k}}{T},
\]
where $\Delta_{k}=t_{k}-t_{k-1}$ indicates the time space between
two consecutive time buckets at $k$-th level. If the time buckets
$t_{k}$ are equally spaced, then the formula reduces to $EPE=\frac{1}{K}{\sum_{k=1}^{K}EE_{k}}$.
Therefore the EPE value gives the time average of the $EE_{k}.$ 
\item We set 
\[
EEE_{1}:=EE_{1}\mbox{ and }EEE_{k}:=Max\left\{ EE_{k},EEE_{k-1}\right\} \;,\mbox{ for every }k=1,\ldots,K.
\]
Due to its non decreasing property, $EEE_{k},$ which is called the
\emph{Effected Expected Exposure,} takes into account the fact that,
once the time decay effect reduces the MtM as well as the counterparty
risk exposure, the bank applies a roll out with some new deals. 
\item We define the \emph{Effected Expected Positive Exposure} (EEPE), by
\[
EEPE:=\frac{\sum_{k=1}^{K}EEE_{k}\cdot\Delta_{k}}{T}\;.
\]

\end{itemize}
In order to avoid too many inessential regulatory details, we will
work on $EE_{k}$ and EPE, the others being just arithmetic modifications
of them.

\begin{remark}\label{remark1} Let us point out that the definition
of $EE_{k}$ is taken from the Basel regulatory framework. We find
it quite strange, since, instead of giving a theoretical principle
and suggesting the Monte Carlo technique just as a possible computational
tool, the simulation approach is officially embedded in the general
definition. \end{remark}

In what follows we shall rewrite previously defined quantities in
continuous time, and we add the index $A$ to indicate the \textit{adjusted}
definitions. Moreover we consider the dynamics of the underlying $S_{t}:=\left\{ S_{t}\right\} _{t\in[0,T]}$,
$T\in\mathbb{R}^{+}$ being some expiration date, as an It\^{o} processes,
defined on some filtered probability space $\left(\Omega,\mathcal{F},\mathcal{F}_{t\in[0,T]},\mathbb{P}\right)$.
As an example, $S_{t}$ is the solution of the stochastic differential
equation defining the geometric Brownian motion, $\mathcal{F}_{t\in[0,T]}$
is the natural filtration generated by a standard Brownian motion
$W_{t}=(W_{t})_{t\in[0,T]}$ starting from a complete probability
space $\left(\Omega,\mathcal{F},\mathbb{P}\right)$, where $\mathbb{P}$
could be the so called \textit{real world} probability measure or
an equivalent risk neutral measure in a martingale approach to option
pricing, see, e.g., \cite[Ch.5]{Shreve}.

The \emph{Adjusted Expected Exposure} $EE^{A}$ is given by 
\begin{equation}
EE_{k}^{A}:=\mathbb{E}_{\mathbb{P}}\left[MtM\left(t_{k},S_{k}\right)^{+}\right]=\int\varphi\left(t_{k},S_{k},\Theta\right)d\mathbf{\mathbb{P}}\cong\frac{1}{N}\sum_{n=1}^{N}MtM\left(t_{k},S_{k,n}\right)^{+}=\widehat{EE_{k}^{A}}\label{eeA}
\end{equation}

We define the \emph{Adjusted Expected Positive Exposure} $EPE^{A}$
as 
\begin{equation}
EPE^{A}:=\int EE_{t}^{A}\de t=\int\left[\int\varphi\left(t,S_{k},\Theta\right)\de\mathbb{P}\right]\de t\;.\label{epeA}
\end{equation}

In this new formulation, the Basel definition is simply one of the
many methods to estimate the expected fair value of the derivative
in the future.

\begin{remark}\label{remark2} We skip any comment about the choice
of the most suitable probability measure $\mathbf{\mathbb{P}}$ to
be used in the calculation of $EE_{k},$ this being beyond the aim
of this paper.

For a detailed discussion on the role played by the \emph{risk neutral}
probability for the drift $S_{t}$ or the \emph{historical} real world
probability, see, e.g., \cite{cas}. \end{remark}

\begin{remark}\label{remark3} Let us observe that the discount factor,
or \emph{numeraire}, is missing in the EPE definition. It was not
forgot, but this is one of the several conservative proxies used in
the risk regulation. \end{remark}

If we adopt a simulation approach, for the underlying path construction
we could generate, for each simulation \emph{n}, a path with an array
of points $\left(x_{n,t_{k}}\right)$. This algorithm is called \emph{path-dependent}
\emph{simulation} (PDS). Alternatively, for each time bucket and for
each simulation, we could jointly generate our $N\cdot K$ points.
This approach is referred as \emph{direct-jump to simulation} date
(DJS). We will come back on PDS and DJS approaches in Section \ref{num}.

The figures below, taken from \cite{pyk1}, well clarify the difference

\begin{figure}[H]
\centering{}\includegraphics[scale=0.1]{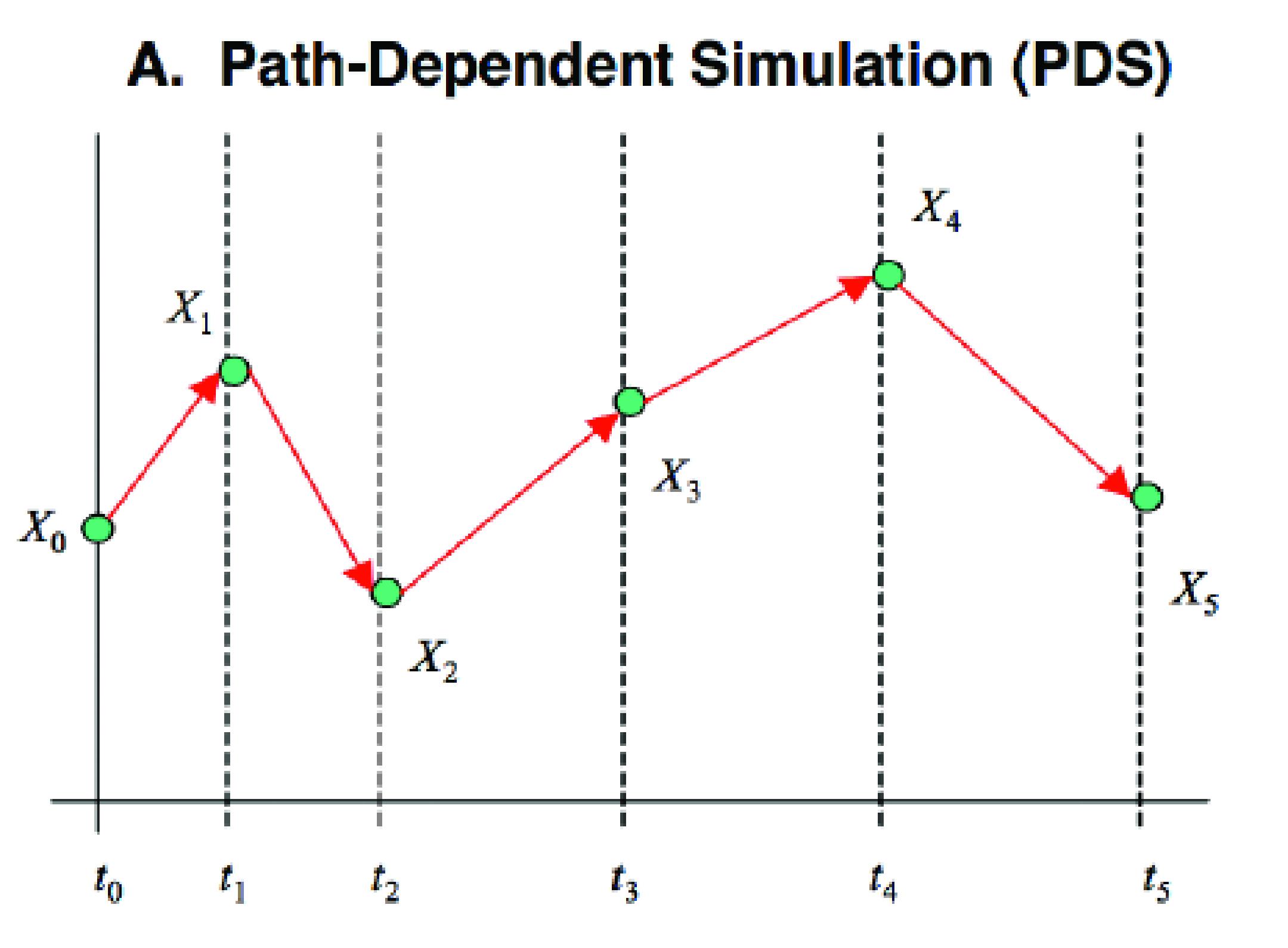}\includegraphics[scale=0.1]{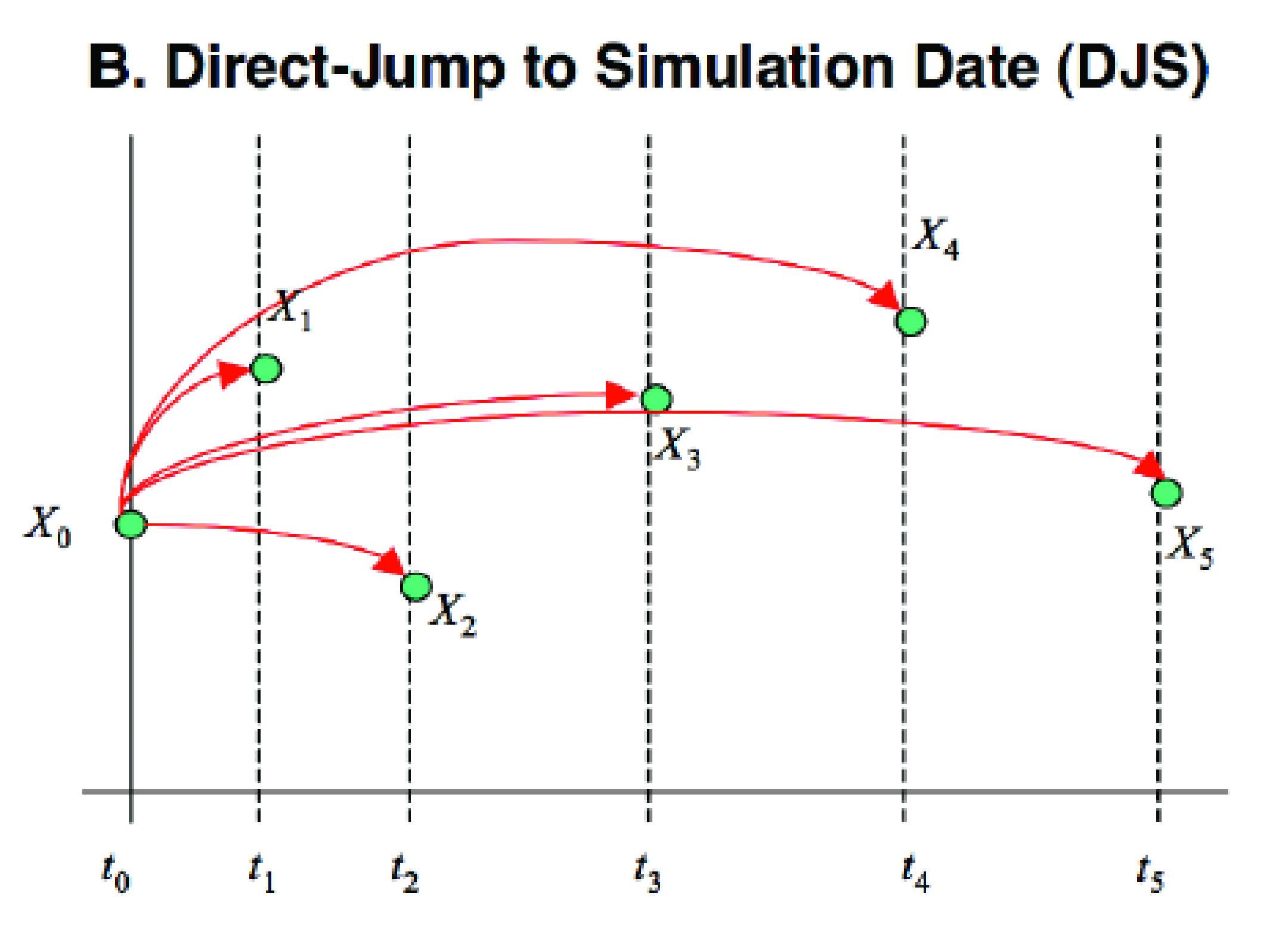}\label{PDS_DJS}
\protect\protect\caption{{\footnotesize{}{}{On the left: an example of PDS approach with
six time buckets. On the right: an example of DJS approach with six
time buckets.}}}
\end{figure}

Finally, we recall that, in the risk management application, another
widely used quantity is the \emph{potential future
exposure} $PFE_{\alpha},$ a quantile based figure of the extreme
risk. In a continuous setting, we define the potential future exposure
in the following way 
\begin{equation} \label{PFE}
PFE\left(\alpha,t_{k}\right) := MtM^{*}\mbox{ such that }P\left\{ MtM\left(t_{k},S_{k}\right)^{+} 
\geq MtM^{*}\right\} =1-\alpha\;.
\end{equation}

\section{A short quantization review} \label{q_review}

The \emph{quantization} technique has been known from several decades
and it comes from engineering, when addressing the issue of converting
an analogical signal, e.g. images or sounds, into a discretized digital
information. Other important areas of application are data compression
and statistical multidimensional clustering. For a classical reference
concerning the quantization approach, we refer to \cite{ger}, while 
\cite{pag2} gives a survey of the literature concerning
fair value \emph{pricing} problems for plain vanilla options, American
and exotic options, basket CDS.

In addition, alternative quantization approaches, such as the so-called
\emph{dual quantization} and the treatment of underlying assets driven
by more structured stochastic processes, are taken into consideration
in \cite{pag4} and \cite{pag3}.

In this section, we shall give a sketch of the quantization idea,
by emphasizing its practical features, but without giving all the details
concerning the mathematical theory behind it.

Let $X\in\mathbb{R}^{d}$, $d\in\mathbb{N}^{+}$, be a $d-$dimensional
continuous random variable, defined over the probability space $\left(\Omega,\mathscr{\mathcal{F}},\mathbb{P}\right)$
and let $\mathbb{P}_{X}$ the measure induced by $X$. The quantization
approach is based on the \textit{approximation} of $X$ by a $d$-dimensional
discrete random variable $\widehat{X}$, further details of which
will be given later, defined by means of a so called \textit{quantization
function} $q$ of $X$, that is to say, $\widehat{X}:=q\left(X\right),$
in such a way that $\widehat{X}$ takes $N\in\mathbb{N}^{+}$ finitely
many values in $\mathbb{R}^{d}.$ The finite set of values for $\widehat{X}$
is denoted by $q\left(\mathbb{R}\right):=\left\{ x_{1},...,x_{N}\right\} $
and it is called a \emph{quantizer} of $X,$ while the application
$q\left(X\right)$ is the related \emph{quantization}. To distinguish
the one-dimensional case ($d=1$) from the $d-$dimensional one ($d>1$),
the terms \emph{quantization}, resp. \emph{vector quantization} (VQ),
are usually used.

Such a set of points in $\mathbb{R}^{d}$ can be used as \textit{generator
points} of a \emph{Voronoi tessellation.} Let us recall that, if ${\bf {X}}$
is a metric space with a distance function ${\bf {d}},$ ${\bf {K}}$
is an index and $(P_{k})_{k\in{\bf {K}}}$ is a tuple of ordered collection
of nonempty subsets of ${\bf {X}}$, then the \emph{Voronoi cell}
${\bf {R}_{k}}$ generated by the site $P_{k}$ is defined as the following the set 
\[
{\bf {R}}_{k}:=\left\{ x\in X\;|\;{\bf d}(x,P_{k})\leq{\bf d}(x,P_{j})\;\text{for all}\; j\neq k\right\} \;.
\]
Therefore, the Voronoi tessellation is the tuple of cells $\left({\bf R}_{k}\right)_{k\in{\bf {K}}}$.
In our case such a tessellation reduces to substitute the set of cells
$P_{k}$ with a finite number $x_{1},\ldots,x_{N}$ of distinct points
in $\mathbb{R}^{d}$, so that the Voronoi cells are convex polytopes.

More precisely, we construct the following Voronoi tessels with respect
to the euclidean norm $\|\cdot\|$ in $\mathbb{R}^{d}:$ 
\[
C\left(x_{i}\right)=\left\{ y\in\mathscr{\mathcal{\mathfrak{\mathbb{R^{\mathit{d}}}}\mathit{\mathrm{:}\left|y-x_{i}\right|<\|y-x_{j}\|\forall j\neq i}}}\right\} \;,
\]
with associated quantization function $q$ defined as follows 
\begin{equation}
q\left(X\right)=\sum_{i=1}^{N}x_{i}\cdot1_{C_{i}\left(x\right)}\left(X\right)\;.\label{qfunction}
\end{equation}

Such a construction allows us to rigorously define a probabilistic
setting for the claimed random variable $\widehat{X},$ by exploiting
the probability measure induced by the continuous random variable
$X.$ In particular, we have a probability space $\left(\Omega,\mathcal{F},\mathbb{P}_{\widehat{X}}\right)$,
where the set of elementary events is given by $\Omega:=\left\{ x_{1},\ldots,x_{N}\right\} $,
and the probability measure $\mathbb{P}_{\widehat{X}}$ is defined
by the following set of relations 
\[
0<\mathbb{P}_{\widehat{X}}\left(x_{i}\right):=P_{X}\left(X\in C\left(x_{i}\right)\right)=:p_{i},\;\mbox{for }i=1,\ldots,N\;.
\]

The goal of such an approximation is to deal efficiently with applications
that arise when calculating some functionals of the random vector
$X$, as in the derivative pricing problem case, in order to evaluate the
expectation $\mathbb{E}\left[f\left(X\right)\right]$ of a certain
payoff function $f$ of $X$ or when we have to deal with a quantile
base indicator, as it happens in the risk management field.

We would like to take into account the former case, in particular we will consider 
the following approximation
\[
\mathbb{E}\left[f\left(X\right)\right]\cong\mathbb{E}\left[f\left(\widehat{X}\right)\right]=\sum_{i}f\left(x_{i}\right)\cdot p_{i}\;,
\]
with control on the accuracy of the chosen quantization.

Let us Assume $X\,\in\, L^{p}\left(\Omega,\mathrm{\mathscr{\mathcal{F}}},\mathbb{P}\right),p\,\in\,\left(1,\infty\right)$
and define the $L^{p}$ error as follows 
\begin{equation}
E\left[\|X-\widehat{X}\|^{p}\right]:=\int_{\mathbb{R}^{d}}\mathrm{\min_{i=1,\ldots,N}}\|x-x_{i}\|^{p}\de\mathbb{P}_{X}\left(x\right)\;,\label{LpError}
\end{equation}
where we denote by $d\mathbb{P}_{X}$ the probability density function
characterizing the random variable $X.$ The integrand in eq. \eqref{LpError}
is always well defined, being a minimum with respect to the finite
set of generators $x_{1},\ldots,x_{N}.$

Concerning eq. \eqref{LpError}, 
the task is to find $the$, or, at least, one, \emph{optimal quantizer}, understood
as the set of Voronoi generators minimizing the value of the integral,
once both parameters $d,N,p$ and the probability density of $X$
are given. Even if such a problem could be particularly difficult in
the general case, also because  it may rise to infinitely many solutions, this
is not the case in our setting. In fact,  we aim at considering a standard Black-Sholes
framework, where the only source of randomness is given by a Brownian
motion. In particular, we shall deal with the pricing problem related
to a European style option, therefore we are interested in the distribution
of the driving random perturbation at maturity time $T$, which means
that we are dealing with the quest of an optimal quantizer for a $d$-dimensional
Gaussian random variable, assumed to be standard, up to suitable transformation
of coordinates, namely $X\sim\mathcal{N}\left(0,I_{d}\right)$.

Algorithms to get the optimal quantizer can be found within the aforementioned references.
Moreover, when the dimension  $d\leq10,$ there is a well established literature concerning
how to find the optimal quantizers when the Gaussian framework is considered,
see, e.g., the web site \textit{http://www.quantize.maths-fi.com/} and \cite{pag1,pag2,Sellami}.

A particularly important quantity related to the choice of the optimal
quantizer is represented by the so called \emph{distortion} parameter
\begin{equation}
D_{X}^{N}(\widehat{X}):=E\left[\|X-\widehat{X}\|^{2}\right]=\int_{\mathbb{R}^{d}}\mathrm{\min_{i=1,\ldots,N}}\|x-x_{i}\|^{2}\cdot d\mathbb{P}_{X}\left(x\right)\;,\label{Distortion}
\end{equation}
which is defined with respect to the quantizer $\left\{ x_{1},\ldots,x_{N}\right\} .$

If the quantization function $\widehat{X}$ takes values in the set
of optimal generators, then the distortion parameter admits a minimum, which will be indicated by $\underline{D}_{X}^{N},$
with $lim_{N\rightarrow\infty}\underline{D}_{X}^{N}=0.$

The following theorem, originally due to Zador, see \cite{Zad1,Zad2},
generalized by Bucklew and Wise in  \cite{BucWise} and then revisited
in \cite{pag2} in its non asymptotic version as a reformulation
of the Pierce lemma, gives a quantitative result about the distortion
magnitude.

\begin{theorem}{{[}Zador{]}} \label{zador_thm} Let $X\in L^{p+\varepsilon}\left(\Omega,\mathcal{F},\mathbb{P}\right)$,
for $p\in(0,+\infty)$, $\varepsilon>0$, and let $\Gamma$ be the
$N-$size tessellation of $\mathbb{R}^{d}$ related to the quantizer
$\widehat{X}.$ Then, 
\begin{equation}
\mathrm{\underset{\mathit{N}}{lim}}\left(N^{\frac{p}{d}}\mathrm{\underset{|\Gamma|\leq N}{\min}}\|X-\widehat{X}\|_{\mathbb{P}}^{p}\right)=J_{p,d}\left(\underset{\mathbb{R}^{d}}{\int}g^{\frac{d}{d+p}}\left(x\right)\de x\right)^{1+\frac{p}{d}}\;,
\end{equation}
where we assume $\de\mathbb{P}=g\de\lambda_{d}+\de\nu,$ for some
suitable function $g,$ and $\nu\perp\lambda_{d},$ $\de\lambda_{d}$
being the Lebesgue measure on $\mathbb{R}^{d}$, while the constant
$J_{p,d}$ corresponds to the case $X\sim Unif\left([0,1]^{d}\right).$
\end{theorem}

Let us recall that the optimal quantizer is \emph{stationary} in the
sense that $E\left[X|\widehat{X}\right]=\widehat{X}$, hence $\int_{C\left(x_{i}\right)}\left(x_{i}^{*}-x\right)\de\mathbb{P}_{X}\left(x\right)=0,\: i=1,\ldots,N.$
In what follows, we focus on the case $d=1$ and $p=2,$ hence
considering a one-dimensional stochastic process and the quadratic distortion
measure, therefore, in terms of Th. \eqref{zador_thm}, we have $J_{p,1}=\frac{1}{2^{p}\cdot\left(p+1\right)}$,
hence $J_{2,1}=\frac{1}{12}.$

\begin{remark}\label{remark5} For practical applications, and in
order to compare numerical results obtained by quantization with those
produced by standard Monte Carlo techniques, we are mainly interested
in the order of convergence to zero of the distortion parameter. In
particular, thanks to Zador Theorem, we have that the quadratic distortion
is of order $O\left(N^{-2}\right).$ \end{remark}

It is also possible to provide results concerning the accuracy of
the approximation $\mathbb{E}\left[f\left(\widehat{X}\right)\right],$
by mean of the distortion's properties, see \cite{pag2}. In particular,
we can distinguish the following cases: 
\begin{description}
\item [{Lipschitz case}] if $f$ is assumed to be a Lipschitz function,
then 
\begin{equation}
\left|E\left[f\left(X\right)\right]-E\left[f\left(\widehat{X}\right)\right]\right|\leq K_{f}\cdot\left\Vert X-\widehat{X}\right\Vert _{1}\leq KL_{f}\cdot\left\Vert X-\widehat{X}\right\Vert _{2}\;.\label{lip1}
\end{equation}

\item [{The smoother Lipschitz derivative case}] If $f$ is assumed to
be continuously differentiable with Lipschitz continuous differential
$Df$, then, by performing the quantization using  an optical
quadratic grid $\Gamma$ and by applying Taylor expansion, we have
\begin{equation}
\left|E\left[f\left(X\right)\right]-E\left[f\left(\widehat{X}\right)\right]\right|\leq KL_{Df}\cdot\left\Vert X-\widehat{X}\right\Vert _{1}\;.\label{lip2}
\end{equation}

\item [{The Convex Case}] If $f$ is a convex function and $\widehat{X}$
is stationary, then exploiting the Jensen inequality, we have 
\begin{equation}
E\left[f\left(\widehat{X}\right)\right]=E\left[f\left(E\left[X|\widehat{X}\right]\right)\right]\leq E\left[f\left(X\right)\right]\;,\label{convex}
\end{equation}

\end{description}
hence, the approximation by the quantization is always a lower bound
for the true value of $E\left[f\left(X\right)\right].$

\begin{remark}\label{remark optimal grid} We emphasize that the
(optimal) quantization grid for given parameter values of $N,d$ and
$p$ is calculated \emph{off-line} once and for all. Then, in the
computational effort comparison vs. a strict sense Monte Carlo approach,
we must take into account that with the quantization one only has
to \emph{plug-in} the points in the numerical model, not to calculate
or simulate them. \end{remark}

\begin{remark}\label{remark funct} An increasing literature is devoted
to the \emph{functional quantization}. In this case, the ``random
variable'', which has to be discretized in an optimal way, belongs
to a suited functional space, e.g. one can consider  an application $X$ such that $X:\Omega\rightarrow\left(L_{T}^{2},\left\Vert \cdot \right\Vert _{2}\right),$
where $L_{T}^{2}=L_{T}^{2}\left(\left[0,T\right],\de t\right).$ We
recall that, even if in mathematical finance applications the stochastic
calculus in continuous time is a very useful tool, in practice we
have to deal with discrete \emph{sampling}, in fact, Asian options
or any other \emph{look-back} derivatives have to work with a discrete
calendar for the \emph{fixing} instants. Then, depending on the specific
application, one can choose if to approximate the discrete time real-world-problem
by optimal quantization or if it is better to quantize the continuous
time setting and then to apply it to the practical application, see,
e.g.,\cite{pag4} for a survey on such a topic. \end{remark}

\section{The Proposal: \emph{quantization} for the EPE calculation}

\label{proposal}

In the following, we focus our attention on the calculation of EPE
for option styles derivatives in the Black-Scholes standard setting,
see \cite{bs}. Strictly speaking, the underlying evolves according
to the following stochastic differential equation

\begin{equation}\label{BS}
\de S_{t}  =S_{t}\cdot\de t+\sigma S_{t}\de W_{t}\;,
\end{equation}
with  related solution
\begin{equation}\label{BSsolution}
S_{t}  =S_{0} \cdot \exp\left[\left(r-\frac{\sigma^{2}}{2}\right)t+\sigma W_{t}\right]\;,
\end{equation}
where $r,\sigma>0,$ $(W_{t})_{t\geq0}$ is a Brownian motion, while $S_{0}$
is the initial value of the underlying $S_{t}.$

It is well known that the plain vanilla \emph{call} and \emph{put}
options have a closed pricing formula. Since we do not want to give here a survey
on the several extensions to the model, we content ourselves  saying that,
in the equity and Forex derivatives markets, the most important model
extensions of eq. \eqref{BS} are the \emph{local volatility models,} the \emph{Heston}
model and the \emph{SABR} models, see, e.g., \cite{Dup}, \cite{Hes},
\cite{Hag}, respectively.

As usual in a new methodology proposal, as a first step we prefer
to check the techniques in a simple framework, in order to have clear
insights about its properties.

We guess that the Monte Carlo approach is just one of the many feasible
techniques for EE and EPE calculation. After all, it is computationally quite
expensive, as shown by the following example.
A medium bank easily has $D=O\left(10^{4}\right)$ derivatives
deals. In order to validate the internal models for the EPE calculation,
usually the central banks require at least $K=20$ time steps and
$N=2000$ simulations. Finally, let us suppose that the relevant underlying
(often called \emph{risk factors}) to be simulated have $O\left(10^{3}\right)$
order. It is the sum of equity underlying, FX significant rates and
rate curve points. Let $U$ be such a parameter.

What about the computational effort for an EPE process task on
the whole book? If we adopt a pure Monte Carlo strategy, we must distinguish
between the two main steps: 
\begin{enumerate}
\item simulation (and storage) of the underlying paths 
\item evaluation of the EE quantities. We omit for simplicity the last
EPE layer, since it is just an algebraic recombination of the EEs. 
\end{enumerate}
The first step implies a \emph{grid} of $G=K\cdot N\cdot U=4\cdot10^{7}$
points, which work as an input for the step 2. This one requires $NT=K\cdot N\cdot D=4\cdot10^{8}$
different tasks. By recalling definition of $EE_{k},$ each of these
tasks is a \emph{pricing process,} which very often  turns to be 
performed by a Monte Carlo algorithm with several thousands of simulation
steps.

Generally speaking, the evaluation of EPE by rough Monte Carlo
is $K\cdot N=O\left(10^{4}\right)$ more expensive than the usual
daily \emph{end-of-day} mark to market evaluation of the book.

Hence, we argue that the brutal Monte Carlo approach can not be a
satisfactory way for the EPE calculation.

For this reason, some banks are exploiting some innovative technological
approaches, such as the use of the \emph{graphical processing unit
}(GPU), instead of CPUs, to set up a parallel calculation system, with
some new programming languages such as NVIDIA, while some other banks have
been invested a lot buying \emph{grid computing} platforms.

We believe that an algorithmic based improvement could be more efficient
than the hardware innovation, or it can be combined with, and much
less expensive. In the derivatives pricing field, such a mixed approach
is well explained in \cite{pag5}.

Coming back to our credit exposure estimation, we try to figure how
to use the quantization technique. At a first level, we can distinguish
between \emph{path-dependent} derivatives and \emph{non} path-dependent
derivatives, in the following \emph{pd} and \emph{npd} for brevity.
We point out that this definition is different from the usual \emph{plain}
vs. \emph{exotic} derivatives. Among the non path-dependent derivatives
we include not only the plain European options, but also European
and American style options with exotic payoff, e.g. mixed digital continuous,
spread options, etc. In the \emph{npd} class, we will work with
the \emph{Asian} options, probably the most popular one.

For the \emph{npd} derivatives, the quantization for the EPE
simply reduces to set up the problem by selecting the parameters $\left(N_{k}\right)_{k=1,\ldots,K}$
for the quantization size at each time bucket $\left(t_{k}\right)$
and then to compute the EPE quantized approximation. We use the
optimal quantizer case, recalling that 
\[
S_{0}\exp\left[\left(r-\frac{\sigma^{2}}{2}\right)t+\sigma\cdot W_{t}\right] = 
S_{0}\exp\left[\left(r-\frac{\sigma^{2}}{2}\right)t+\sigma\cdot\sqrt{t}\cdot N\left(0,1\right)\right]\;.
\]

More formally, if we indicate by the exponent \emph{Q} the quantized
EPE, one easily gets

\begin{align}
EE_{k}^{Q} & =\sum_{i=1}^{N_{k}}MtM\left(t_{k},S\left(x_{i}^{k}\right)\right)^{+}p_{i}^{k}\label{eq_EE_Q}\;,\\
EPE^{Q} & =\frac{\sum EE_{k}^{Q}\cdot\Delta_{k}}{T}=\frac{\underset{k}{\sum}\Delta_{k}\left(\sum_{i=1}^{N_{k}}MtM\left(t_{k},S\left(x_{i}^{k}\right)\right)^{+}\cdot p_{i}^{k}\right)}{T}\;.\label{eq_EPE_Q}
\end{align}

Fig. \eqref{evolutionEEQ}, graphically explains the calculation procedure. In particular, the black point
on the left is the underlying level at time $t_{0}$. For each time
step $t_{k}$ and for each point of the quantizer $x_{i}^{k}$, a
MtM is calculated and it is weighted by the probability $p_{i}^{k}.$
Hence $EE^{Q}$ and $EPE^{Q}$ straightly follow.

\begin{figure}[H]\label{evolutionEEQ}
\centering{}\includegraphics[scale=0.6]{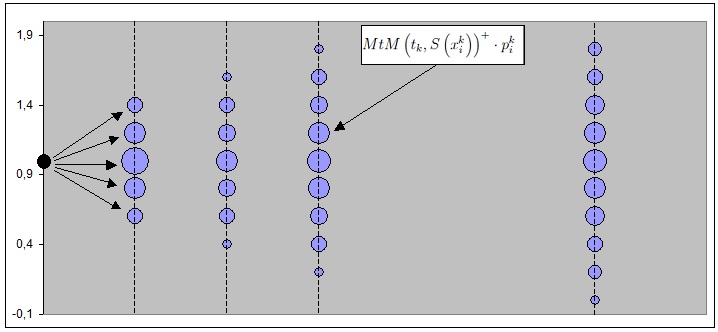} \protect\protect\caption{{\footnotesize{}{}{evolution of $EE^{Q}$ at every time bucket $t_{k}.$}}}
\end{figure}

Concerning the theoretical properties of such an approximation, we
provide a useful result, which can be easily applied to the case of
European style options.

\begin{proposition}\label{CCR BlackScholes} Let us consider an option
in the Black and Scholes setting, with $d=1$ and suppose that the
pricing function $\varphi\left(\cdot\right)$ is Lipschitz continuous or
continuously differentiable with a Lipschitz continuous differential. 

Moreover, let us define $D_{N}^{EPE}:=\left|EPE^A-EPE^Q\right|^2,$
where $EPE^{A}$ is the Adjusted Expected Positive Exposure, defined
in eq. \eqref{epeA}, and $EPE^{Q}$ is the quantized Expected Positive
Exposure, defined in eq. \eqref{eq_EPE_Q}. Then, we have 
\begin{align*}
D_{N}^{EPE} & \propto N^{-2}\cdot K^{-1},\mbox{ if \ensuremath{\varphi}Lipschitz continuous}\;,\\
D_{N}^{EPE} & \propto N^{-4}\cdot K^{-1},\mbox{ if \ensuremath{\varphi}cont. differentiable with Lipschitz cont. differential}\;.
\end{align*}
\end{proposition}

\begin{proof} By definition, we have 
\[
\left|EPE^{A}-EPE^{Q}\right|^{2}=\left|\frac{\sum EE_{k}^{A}\Delta_{k}}{T}-\frac{\sum EE_{k}^{Q}\Delta_{k}}{T}\right|^{2}=\frac{1}{T^{2}}\left|\sum\Delta_{k}\cdot\left(EE_{k}^{A}-EE_{k}^{Q}\right)\right|^{2}\;,
\]

hence rearranging the terms and  recalling that the CCR is effective
just for the \emph{buy} side position, we can skip the positive part
obtaining 

\begin{align}
\left|\sum\Delta_{k}\left(EE_{k}^{A}-EE_{k}^{Q}\right)\right|^{2} & =\left|\sum\Delta_{k}\cdot\left(\mathbf{E}_{\mathbf{P}}\left[MtM\left(t_{k},S_{k}\right)\right]-\mathbf{E}_{\mathbf{\widehat{P}}}\left[MtM\left(t_{k},S_{k}\right)\right]\right)\right|^{2}\nonumber \\
 & \leq\sum\Delta_{k}^{2}\left(\mathbf{E}_{\mathbf{P}}\left[MtM\left(t_{k},S_{k}\right)\right]-\mathbf{E}_{\mathbf{\widehat{P}}}\left[MtM\left(t_{k},S_{k}\right)\right]\right)^{2}\;.\label{eq1}
\end{align}

Moreover, we have 
\[
\left|\sum\Delta_{k}\left(EE_{k}^{A}-EE_{k}^{Q}\right)\right|^{2}\leq\left(\sum\Delta_{k}\left|\mathbf{E}_{\mathbf{P}}\left[MtM\left(t_{k},S_{k}\right)\right]-\mathbf{E}_{\mathbf{\widehat{P}}}\left[MtM\left(t_{k},S_{k}\right)\right]\right|\right)^{2}\;.
\]

In a more explicit fashion, let us work on the single element $k$
of the summation, i.e. $\left(\mathbf{E}_{\mathbf{P}}\left[MtM\left(t_{k},S_{k}\right)\right]-\mathbf{E}_{\mathbf{\widehat{P}}}\left[MtM\left(t_{k},S_{k}\right)\right]\right).$
If we consider $MtM\left(t_{k},S_{k}\right)=\varphi\left(t_{k},S_{k}\left(X\right)\right)$
as the function appearing in eq. \eqref{lip1} and eq. \eqref{lip2},
we get, respectively,

\begin{align}
\left(\mathbf{E}_{\mathbf{P}}\left[MtM\left(t_{k},S_{k}\right)\right]-\mathbf{E}_{\mathbf{\widehat{P}}}\left[MtM\left(t_{k},S_{k}\right)\right]\right)^{2} & \leq KL_{f,k}^{2}\cdot\left\Vert X_{k}-\widehat{X_{k}}\right\Vert _{1}^{2}\label{stima1}\\
\left|\mathbf{E}_{\mathbf{P}}\left[MtM\left(t_{k},S_{k}\right)\right]-\mathbf{E}_{\mathbf{\widehat{P}}}\left[MtM\left(t_{k},S_{k}\right)\right]\right| & \leq KL_{Df,k}^{2}\cdot\left\Vert X_{k}-\widehat{X_{k}}\right\Vert _{2}^{2}\;.\label{stima2}
\end{align}

By replacing eq. \eqref{stima1} and eq. \eqref{stima2} in eq. \eqref{eq1},
we obtain 
\begin{align}
\left|\sum\Delta_{k}\left(EE_{k}^{A}-EE_{k}^{Q}\right)\right|^{2} & \leq\sum\Delta_{k}^{2}KL_{f,k}^{2}\cdot\left\Vert X_{k}-\widehat{X_{k}}\right\Vert _{1}^{2}\;,\label{stima3}\\
\left|\sum\Delta_{k}\left(EE_{k}^{A}-EE_{k}^{Q}\right)\right|^{2} & \leq\sum\Delta_{k}^{2}KL_{Df,k}^{2}\cdot\left\Vert X_{k}-\widehat{X_{k}}\right\Vert _{2}^{4}\;.\label{stima4}
\end{align}

Let us consider eq. \eqref{stima3}, the calculation for eq. \eqref{stima4}
being the same.

Set $\overline{KL}$ equal to the mean of $\left(KL_{\cdot,k}\right),$
for all $k,$ and suppose that the mesh $t_{k}$ is regular enough,
i.e. we require that $\mathrm{lim_{K}}\Delta_{k}=0,\, O\left(\Delta_{k}\right)=K^{-1}\;\forall\, k.$
Thanks to the Zador Theorem and taking $N\rightarrow +\infty$, we have

\begin{align*}
\frac{1}{T^{2}}\cdot\left|\sum\Delta_{k}\left(EE_{k}^{A}-EE_{k}^{Q}\right)\right|^{2} & \leq\frac{1}{T^{2}}\sum\Delta_{k}^{2}KL_{f,k}^{2}\left\Vert X_{k}-\widehat{X_{k}}\right\Vert _{2}^{2}\\
 & \propto\frac{1}{T^{2}}\cdot K\cdot\overline{KL}\cdot\frac{T^{2}}{K^{2}}\cdot N^{-2}\;.
\end{align*}

By simplifying, we get the result. Similar calculations provide the
result  when the pricing function is assumed to be continuously differentiable with Lipschitz continuous differential. 
For a more abstract setting, see \cite[Sec.2.4]{pag2}. \end{proof}

\begin{remark}\label{remark application} Let us discuss the hypothesis
under which the result holds. The central role is played by the function
$\varphi\left(x,\cdot\right)$ as a function of the Brownian motion
$W_{t}$, that is, of the quantity $x\cdot\sqrt{t}$, $x$ sampled
from the $N\left(0,1\right).$ We recall in fact that the usefulness
of the quantization is to infer the properties of the approximation
in the specific problem, starting from the Gaussian optimized discretization.
As an example, for a \emph{put} option the pricing function is bounded,
Lipschitz continuous and twice differentiable, since the Black-Scholes
formula is $C^{\infty}.$ Then, the above proposition holds.

For a \emph{call} style option, the \emph{convexity} properties easily
holds, hence the quantization gives us a lower bound to the EPE estimation.
\end{remark}

\begin{remark}\label{remark complexity} If we consider the pricing
function $\varphi\left(\cdot\right)$ as the elementary unit of our
EPE computational work-flow, the computational complexity of
the quantized approach is $\underset{k}{\sum N_{k}},$ to be compared
with the global number of Monte Carlo scenarios simulations. \end{remark}

For path-dependent derivatives, for each time $t_{k},$ at least in
a \emph{discrete sampling} sense, one needs the whole path of the
underlying. Hence, the above approach is not satisfactory, as the
pricing function depends not only on the current level $S_{k},$ but
also on some functions, e.g. the average, min, max, etc., of the
underlying level until $t_{k}$. Latter task can be efficiently studied by an approach like the
\emph{PDS} one, as in figure \ref{PDS_DJS}.

Let $N_{k}$ be a positive integer for the \emph{quantization}, and
$q_{k}\left(\mathcal{R}\right)=\left(x_{1},x_{2,}...,x_{N_{k}}\right)$
is the quantizer of size $N_{k},$ namely the random variable $\widehat{X_{k}}=q\left(X_{k}\right)$
that maps the random variable to an optimal discrete version. If we
refer to the Black-Scholes model, we want to quantize at each step
the Brownian motion $W_{t}$ that generates the log-normal underlying
process. We recall that $W_{t}\sim N\left(0,t\right)$ is a normal
centered random variable and that $W_{t}-W_{s}\sim N\left(0,t-s\right).$

Again, we define $C\left(x_{i}\right)$ as the $i-$th \emph{Voronoi
tessel} such that 
\[
C\left(x_{i}\right):=\left\{ y\in\mathscr{\mathcal{\mathfrak{\mathbb{R^{\mathit{d}}}}\mathit{\mathrm{:}\left|y-x_{i}\right|<\left|y-x_{j}\right|\forall j\neq i}}}\right\} ,
\]
with the so-called \emph{nearest neighborhood principle.} A set of
probability masses vectors is assigned to the $N_{k}-$tuple, let
be $\mathbf{p^{k}}=\left(p_{1}^{k},p_{2}^{k},\ldots,p_{N_{k}}^{k}\right),$
where $p_{i}=\mathbb{P}\left(C\left(x_{i}\right)\right)$ under the
original probability $\mathcal{\mathbb{P_{X}}},$ for all $i=1,\ldots,N_{k}.$
The following questions naturally arise: how and where to use the quantizers for the $EE_{k}$ calculation ? 
The \emph{quantization tree} is a discrete space, discrete time process,
defined by
\begin{align}
p_{i}^{k} & =\mathbf{\mathbb{P}\mathrm{\left(\mathit{\widehat{X}_{k}}=x_{i}^{k}\right)}}=\mathbb{P}\left(X_{k}\in C_{i}\left(x^{k}\right)\right)\label{prob_voronoi} \;,\\
\pi_{ij}^{k} & =P\left(\widehat{X}_{k+1}=x_{j}^{k+1}|\widehat{X}_{k}=x_{i}^{k}\right)=P\left(X_{k+1}\in C_{j}\left(x^{k+1}\right)|X_{k}\in C_{i}\left(x^{k}\right)\right)\nonumber \\
 & =\frac{P\left(X_{k+1}\in C_{j}\left(x^{k+1}\right),X_{k}\in C_{i}\left(x^{k}\right)\right)}{p_{i}^{k}}\label{trans_prob}
\end{align}

The following theoretical result allows us to explicitly solve the
transition probability formula \eqref{qfunction}, by recalling that
$X_{k}$ is the original Brownian motion sampled at a given time $t_{k}.$

\begin{proposition}\label{prop_trans_prob} Let us use denote by
$\Delta_{k}:=t_{k+1}-t_{k}$ the time space between two consecutive
time buckets and by $\phi$ the density of the $N\left(0,1\right)$
random variable. Furthermore, let us indicate by $U_{k},U_{k+1}$, resp. by
 $L_{k},L_{k+1}$, the upper, resp. lower, bounds 
of the $1-$dimensional tessels $C_{i}\left(x^{k}\right)$ and $C_{j}\left(x^{k+1}\right).$

Then, the transition probability $\pi_{ij}^{k}$ assumes the following
expression 
\begin{equation}
\pi_{ij}^{k}=\int_{C_{j,k+1}}f\left(x_{k+1}|C_{i,k}\right)\de x_{k}=\frac{\intop_{L_{k}}^{U_{k}}\left(\intop_{L_{k+1}}^{U_{k+1}}\phi\left(\frac{\left(x-y\right)}{\sqrt{\Delta_{k}}}\right)\de x\right)\phi\left(\frac{y}{\sqrt{t_{k}}}\right)\cdot dy}{p_{i}^{k}}\;.\label{tr}
\end{equation}
\end{proposition}

\begin{proof} The result is a straightforward calculation, indeed let us start
considering a more specific problem, namely we aim at calculate $P\left(X_{k+1}\in C_{j}\left(x^{k+1}\right)|X_{k}=y,\: y\in C\left(x^{k}\right)\right).$

For  given $x\in C_{j}\left(x^{k+1}\right)$, $y\in C_{j}\left(x^{k}\right)$,
by recalling the scaling property and the independence of the Brownian
motion increments, we easily get $P\left(\de x\right)=\phi\left(\frac{\left(x-y\right)}{\sqrt{\Delta}_{k}}\right)\cdot \de x$, and
the result follows by extending such a fact
to tessels $C_{i},C_{j}$.\end{proof}

The picture below shows the practical features of the formula eq.
\eqref{tr}.

\begin{figure}[H]
\centering{}\includegraphics[scale=0.6]{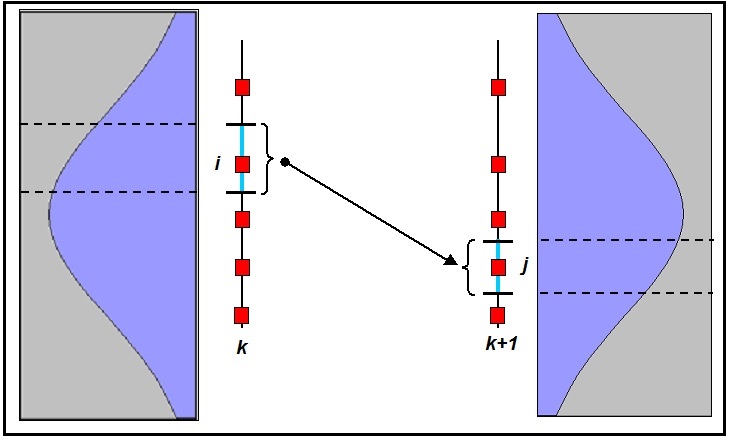} \protect\protect\caption{{\footnotesize{}{}{a graphical representation of the transition
probability function}}}
\end{figure}

\begin{remark}\label{MC} Even if the proposition comes from elementary
probability, this result is a useful improvement to the procedure
in Pag\`es et al (2009), where a Monte Carlo approach for the $\pi_{ij}^{k}$
was suggested. \end{remark}

\begin{remark}\label{cc} From a computational complexity point of
view, we observe that the above coefficients $\pi_{ij}^{k}$ can be
calculated off line, once and for all, given the time discretization
parameter $K$ and the chosen granularities $\left\{ N_{1},N_{2},...,N_{K}\right\} .$
\end{remark}

Despite the reduction in the number of evaluations that the quantization
approach guarantees, the possible paths of the quantization tree,
say $NP,$ could be too many from a theoretical perspective. In fact,
they amount to $NP=\prod_{k}N_{k}$. If we set, as usual, $K=O\left(10^{1}\right)$
and $N_{k}=O\left(10^{2}\right),$ then $NP=O\left(10^{20}\right),$
which seems to be intractable for practical purposes. Fortunately,
this does not occur, from a practical point of view. Many paths have
a negligible probability, as very often we have $\pi_{ij}^{k}\simeq0,$
so we can skip a large fraction of the combinatorial cases by some
heuristic \emph{a priori} rule that avoids calculation depending from
the distance $\left|\left(\frac{x_{j}^{k+1}-x_{i}^{k}}{\sqrt{\varDelta_{k}}}\right)\right|$.

\section{The numerical Application}

\label{num}

In this section we will give an application of quantization method
in CCR with respect to a portfolio consisting of a bank account and
one risky asset which is the underlying of a European type option.
We suppose
that the dynamics of the underlying $S_{t}:=\left\{ S_{t}\right\} _{t\in[0,T]}$,
for some maturity time $T\in\mathbb{R}^{+}$, is given accordingly
to a geometric Brownian motion, namely 
\begin{equation}
\de S_{t}=rS_{t}\de t+\sigma S_{t}\de W_{t}\;,\label{bs_eq}
\end{equation}
where $r$ is the risk free interest rate, e.g. associated to a bank account, $\sigma$ is the volatility parameter characterizing the underlying behaviour, while $W_{t}:=\{W_{t}\}_{t\,\in\,[0,T]}$ is a $\mathbb{R}-$
valued Brownian motion on the filtered probability space $\left(\Omega,\mathcal{F},\mathcal{F}_{t},\mathbb{P}\right)$,
$\left\{ \mathcal{F}_{t}\right\} _{t\in[0,T]}$ being the filtration
generated by $W_{t}$. We recall that the stochastic differential
equation \eqref{bs_eq} admits an explicit solution, see, e.g.,\cite[Ch.3]{Shreve},
given by 
\begin{equation}
S_{t}=S_{0}\exp\left\{ \sigma W_{t}+\left(r-\frac{\sigma^{2}}{2}\right)t\right\} \label{bs_sol} \;,
\end{equation}
$S_0$ being the initial value of the underlying $S_t$.

If we consider a European call option with strike price $K\in\mathbb{R}^{+}$ snd 
maturity time $T$, written on $S_t$ described as above, 
then its fair value $C_{eu}=C_{eu}(S_{0},r,K,\sigma,T)$,
with respect to the unique martingale equivalent measure, is given
by 
\begin{align}
C_{eu}(S_{0},r,K,\sigma,T) & :=e^{-rT}\mathbb{E}\left[\left(S_{T}-K\right)^{+}\right]\nonumber \\
 & =e^{-rT}\mathbb{E}\left[\left(S_{0}\exp\left\{ \left(r-\frac{\sigma^{2}}{2}\right)T+\sigma W_{T}\right\} -K\right)^{+}\right]\;,\label{eu_call}
\end{align}

with explicit solution given by the Black-Scholes formula.

Before enter into details about  a numerical application of our proposal, see Sec. \eqref{proposal}, to a concrete case, let
us underline some points. 
In practical applications, the computational challenges are very
often much harder than one believes from a theoretical perspective.
Referring this general principle to the CCR, the portfolio of derivatives
of a counterparty \emph{A} with \emph{B} may consist of several hundreds
of derivatives $j=1...J$, then the Mark-to-Market is given by $MtM^{A}\left(t\right)=\underset{j}{\sum}MtM_{j}^{A}\left(t\right).$
These derivatives could be \emph{bought} options, s\emph{old} options,
\emph{swaps} and so on. A collateral of value $V_{t}$ is usually
posted. Hence, at the current time, the exposure of the counterparty
\emph{A} to \emph{B} is given by 
\begin{equation} \label{collateral}
\left(\underset{j}{\sum}MtM_{j}^{A}\left(t\right)-V_{t}\right)^{+}\;,
\end{equation}

an  expression which is similar to the one describing a  \emph{call option} written on a
derivatives portfolio. In the CCR applications, one wants to check
several quantities related to the current exposure, such as EE,
EPE, PFE, and so on. In calculating the expected exposure 
of  quantities as in \eqref{collateral}, i.e. $EE_{j}\left(\cdot\right),$ because of non 
linearity, one can not calculate separately the single deal quantities, i.e.
the $EE_{j}\left(t_{k}\right)$, to aggregate them by summation in a second moment.
\emph{A fortiori,} in the PFE quantile framework, one can not infer
easily the quantile  by the specific quantiles.

It follows that all banks, as a general strategy, first calculate
a large set of  scenarios for the underlyings, \emph{coherently}
with respect to the considered probabilistic structure for it, and then they evaluate and store
a large set of \emph{MtM,} from which to pick any kind of statistics
and risk figures. In this field quantization can play a role as a
competitive methodology, particularly to what concerns saving computational costs. 
Nevertheless, since the CCR for a whole book
comes from the single deals computations, we aim to test at a {\it low}
level the quantization, in a plain vanilla context. In further research
we will move to 
Exotic deals as well as effective management of large portfolios, will be the subject of our next studies.

\subsection{Set-up and quantization strategy}

\label{num1}

For the market valuation of financial statements, the generation of
potential market scenarios is required. In Sec. \eqref{BaselEPESection}, we definded
two achievable approaches, namely the path-dependent simulations method
(PDS) and the direct jump to simulation data (DJS) technique: in first case one 
simulates a whole path-wise possible trajectory, while in the second each time point is directly computed.
 

More practically, we could apply the DJS approach by selecting a grid
size $N_{k}$ for each time $t_{k}$ and then using $K$ different
1-dimension quantizers $q_{k}\left({\bf {R}}\right)$.

Alternatively, by using the PDS approach, we can work just with one
$d$-dimension $q\left({\bf {R}^{d}}\right)$ quantizer of cardinality
$N$, hence each dimension works for one time step.

We choose the latter approach for our application. The steps are as
follows. 
\begin{enumerate}
\item Selection of the grid size $N$ and the dimension $d,$ according
respectively to the computational effort constraints and to the time
buckets cardinality $K,$ i.e. $d=K.$ 
\item By obvious dilatation, each point of the $(N\times K)$ grid of the
vector quantizer is mapped in order to get a proper to a Brownian
motion increment realization, i.e. 
\[
x_{i,k}\rightarrow x_{i,k}\sqrt{\Delta_{k}}=\Delta\widetilde{W},\;\forall\, i\,\in\,\{1,\ldots,N\},\, k\,\in\,\{1,\ldots,d\}\,.
\]

\item The above increments are inserted in the Black-Scholes diffusion to
get $N$ ``possible'' underlying paths $S_{t_{k}}.$ 
\item Payoff, MtMs and all the related quantities are calculated,
by using the probability masses $p_{i}.$ 
\end{enumerate}
Although we know that it is not possible to get an exhaustive comparison,
anyhow we try to make the exercise quite general, by choosing some
different parameter combinations, e.g. 
\begin{itemize}
\item spot price $(S_{0}):$ $90,100,110$ 
\item interest rate $r:$ $3\%$ 
\item volatility $\sigma:$ $15\%,25\%,30\%$ 
\item strike price $R:$ $100$ (we do not use the usual $K$ notation to
avoid confusion with the time buckets set $\left\{ t_{k}\right\} $) 
\item time to maturity $T:$ one year 
\end{itemize}
According to the choice of several banks to consider an increasing
sampling frequency over time, because of accuracy decreasing over
large horizons, we decide to set time buckets  in the following
range 
\[
\left\{ 0,\frac{1}{52},\frac{2}{52},\frac{3}{52},\frac{4}{52},\frac{2}{12},\frac{3}{12},\frac{6}{12},\frac{9}{12},1\right\} \;,
\]
namely, we are considering the first, the second, the third and the
fourth week and then the second, the third, the sixth and the twelfth
month of the year.

\subsection{The single option situation}

\label{num2}

In order to evaluate the Expected Positive Exposure (EPE), we compare
the quantization approach with standard Monte Carlo method. We distinguish
several cases, depending on the \emph{moneyness}, i.e. the relative position of $S_{t}$ versus the strike
price $K$ of the considered call option, and volatility parameters.

Each case will be analyzed with respect to the Monte
Carlo-Sobol approach, see, e.g., \cite{Caflish}, with $N=10^{3},$
as well as considering the Monte Carlo simulations (MC) and the quantization grids (Q),
with $N=10^{3}.$

Concerning the choice of the benchmark, let us note that, in the Black-Scholes
setting, in order to price a European call option, we work in a risk
neutral context where the drift of the geometric Brownian motion has
to be equal to the risk free rate. Under such an assumption, the Expected
Exposure (EE) admits a closed form, i.e.

\begin{equation}
EE_{k}^{A}=\mathbf{E}_{\mathbf{P}}\left[MtM\left(t_{k},S_{k}\right)^{+}\right] 
= MtM\left(t_{0},S_{0}\right)\cdot\mathrm{exp}\left(t_{k}-t_{0}\right)\;,\label{benchmark}
\end{equation}

which is implied by  the\emph{ martingale property} of the discounted fair
value, see \cite[Ch.5]{Shreve} for further details.

In a more general setting, the simple approach characterized by eq.\eqref{benchmark} can not be applied, because 
of more involved payoffs. Moreover, the Mark-to-Market function does not exist
in an analytical form and the drift can assume values different from
the risk free rate $r.$ Besides, we are often interested in
calculating a measure of the possible worst exposure with a certain
level of confidence. Such a measure is expressed in terms of $p-$percentile
of the exposure's distribution, the so-called Potential Future
Exposure (PFE), defined in eq. \eqref{PFE}. 

As already stressed, since we  aim at testing the efficiency
of quantization techniques, we refer here to a simple problem, while the case of more complex financial models
is the subject of our next research.

In what follows we always consider a $(D+1)\times1$ matrix, $D$
being the number of time buckets taken into account, hence $D=9,$
since we start from $t_{0}=0.$ Each matrix entry represents the value
of the Expected Exposure (EE), $EE_{k}:=\frac{1}{N}\sum_{n=1}^{N}MtM\left(t_{k},S_{k,n}\right)^{+}$,
obtained by applying eq. \eqref{eq_EE_Q}. Last row gives the value
of the Expected Positive Exposure (EPE), $EPE:=\frac{\sum_{k=1}^{K}EE_{k}\cdot\Delta_{k}}{T},$
calculated according to eq. \eqref{eq_EPE_Q}.

The efficiency evaluation of exploited procedures
requires a comparison between the value obtained by simulations and a benchmark.
In the case of quantization approach, such a comparison is given by
the evaluation of the (percent) relative error $\varepsilon$ with
respect to the Black-Scholes price obtained using formula
\eqref{benchmark}. As regards the Monte Carlo approach, the analyzed
quantity is the (percent) relative standard error (RSD). By the {\it Law of large numbers}, 
it is well known that the Monte Carlo approach always converges
to the true value, hence its standard deviation is more informative
than the single execution error. The {\it numerical} calculation in
the tables stands for the numerical integration of formula \eqref{eeA},
i.e. the expected value of the possibles MtMs over the different
underlying prices. The integration has been done considering a simple rectangle
scheme, with $10^{3}$ points. Finally, we also tested the Monte Carlo-Sobol
technique, based on binary digits that well fill the $[0,1]$ interval.
To summarize, all the different techniques were compared with the
same number of points and avoiding too sophisticated versions, to
keep the comparison as fair as possible.

Tables \ref{tab:ITM15}, \ref{tab:ITM25} and \ref{tab:ITM30}, contain
numerical results in the ITM case with $S_{0}=110,$ while tables \ref{tab:ATM15},\ref{tab:ATM25} and \ref{tab:ATM30},
refer to the ATM case with $S_{0}=100$, finally tables \ref{tab:OTM15},\ref{tab:OTM25}, and \ref{tab:OTM30},
report the performances in the OTM case, with $S_{0}=90.$ 

{\scriptsize{}{}{ } }
\begin{table}[H]
\begin{centering}
{\footnotesize{}{}{{}}}{\scriptsize{}{}}%
\begin{tabular}{|c|c|cc|cc|cc|cc|}
\cline{2-10} 
\multicolumn{1}{c|}{} & {\footnotesize{}{}{{Analytic } }}{\scriptsize{}{}  } & \multicolumn{2}{c|}{{\footnotesize{}{}{{Numerical}}}} & \multicolumn{2}{c|}{{\footnotesize{}{}{{Quantization}}}} & \multicolumn{2}{c|}{{\footnotesize{}{}{{Monte Carlo}}}} & \multicolumn{2}{c|}{{\footnotesize{}{}{{MC-Sobol}}}}\tabularnewline
\hline 
{\footnotesize{}{}{{t } }}{\scriptsize{}{}  } & {\footnotesize{}{}{{EE } }}{\scriptsize{}{}  } & \multicolumn{1}{c}{{\footnotesize{}{}{{EE}}}} & {\footnotesize{}{}{{$\varepsilon$ } }}{\scriptsize{}{}  } & \multicolumn{1}{c}{{\footnotesize{}{}{{EE}}}} & {\footnotesize{}{}{{$\varepsilon$ } }}{\scriptsize{}{}  } & \multicolumn{1}{c}{{\footnotesize{}{}{{EE}}}} & {\footnotesize{}{}{{RSD } }}{\scriptsize{}{}  } & \multicolumn{1}{c}{{\footnotesize{}{}{{EE}}}} & {\footnotesize{}{}{{$\varepsilon$}}}\tabularnewline
\hline 
{\scriptsize{}{}1w  } & {\scriptsize{}{}14,711  } & {\scriptsize{}{}14,710  } & {\scriptsize{}{}-0,007\%  } & {\scriptsize{}{}14,711  } & {\scriptsize{}{}0,000\%  } & {\scriptsize{}{}14,649  } & {\scriptsize{}{}0,004\%  } & {\scriptsize{}{}14,710  } & {\scriptsize{}{}-0,010\% }\tabularnewline
{\scriptsize{}{}2w  } & {\scriptsize{}{}14,719  } & {\scriptsize{}{}14,717  } & {\scriptsize{}{}-0,007\%  } & {\scriptsize{}{}14,719  } & {\scriptsize{}{}0,000\%  } & {\scriptsize{}{}14,725  } & {\scriptsize{}{}0,006\%  } & {\scriptsize{}{}14,717  } & {\scriptsize{}{}-0,014\% }\tabularnewline
{\scriptsize{}{}3w  } & {\scriptsize{}{}14,726  } & {\scriptsize{}{}14,726  } & {\scriptsize{}{}-0,008\%  } & {\scriptsize{}{}14,727  } & {\scriptsize{}{}0,000\%  } & {\scriptsize{}{}14,815  } & {\scriptsize{}{}0,007\%  } & {\scriptsize{}{}14,725  } & {\scriptsize{}{}-0,012\%}\tabularnewline
{\scriptsize{}{}1m  } & {\scriptsize{}{}14,776  } & {\scriptsize{}{}14,734  } & {\scriptsize{}{}-0,008\%  } & {\scriptsize{}{}14,736  } & {\scriptsize{}{}0,000\%  } & {\scriptsize{}{}14,869  } & {\scriptsize{}{}0,008\%  } & {\scriptsize{}{}14,735  } & {\scriptsize{}{}-0,002\%}\tabularnewline
{\scriptsize{}{}2m  } & {\scriptsize{}{}14,813  } & {\scriptsize{}{}14,774  } & {\scriptsize{}{}-0,010\%  } & {\scriptsize{}{}14,775  } & {\scriptsize{}{}0,000\%  } & {\scriptsize{}{}15,003  } & {\scriptsize{}{}0,012\%  } & {\scriptsize{}{}14,775  } & {\scriptsize{}{}-0,005\%}\tabularnewline
{\scriptsize{}{}3m  } & {\scriptsize{}{}14,924  } & {\scriptsize{}{}14,811  } & {\scriptsize{}{}-0,011\%  } & {\scriptsize{}{}14,812  } & {\scriptsize{}{}0,000\%  } & {\scriptsize{}{}14,801  } & {\scriptsize{}{}0,014\%  } & {\scriptsize{}{}14,808  } & {\scriptsize{}{}-0,030\%}\tabularnewline
{\scriptsize{}{}6m  } & {\scriptsize{}{}15,036  } & {\scriptsize{}{}14,921  } & {\scriptsize{}{}-0,016\%  } & {\scriptsize{}{}14,924  } & {\scriptsize{}{}0,000\%  } & {\scriptsize{}{}14,916  } & {\scriptsize{}{}0,020\%  } & {\scriptsize{}{}14,924  } & {\scriptsize{}{}-0,001\%}\tabularnewline
{\scriptsize{}{}9m  } & {\scriptsize{}{}15,149  } & {\scriptsize{}{}15,033  } & {\scriptsize{}{}-0,019\%  } & {\scriptsize{}{}15,036  } & {\scriptsize{}{}0,000\%  } & {\scriptsize{}{}15,157  } & {\scriptsize{}{}0,026\%  } & {\scriptsize{}{}14,994  } & {\scriptsize{}{}-0,282\%}\tabularnewline
{\scriptsize{}{}1y  } & {\scriptsize{}{}15,149  } & {\scriptsize{}{}15,145  } & {\scriptsize{}{}-0,023\%  } & {\scriptsize{}{}15,149  } & {\scriptsize{}{}-0,001\%  } & {\scriptsize{}{}14,870  } & {\scriptsize{}{}0,031\%  } & {\scriptsize{}{}15,049  } & {\scriptsize{}{}-0,659\%}\tabularnewline
\hline 
{\scriptsize{}{}EPE  } & {\scriptsize{}{}14,970  } & {\scriptsize{}{}14,967  } & {\scriptsize{}{}-0,017\%  } & {\scriptsize{}{}14,970  } & {\scriptsize{}{}0,000\%  } & {\scriptsize{}{}14,951  } & {\scriptsize{}{}-0,125\%  } & {\scriptsize{}{}14,934  } & {\scriptsize{}{}-0,241\%}\tabularnewline
\hline 
\end{tabular}
\par\end{centering}{\scriptsize \par}

{\footnotesize{}{}{{\protect\protect\caption{{\footnotesize{}{}{{EPE: $10\%-$ITM European call. $\sigma=15\%$.\label{tab:ITM15}}}}}
}{\footnotesize \par}

{\footnotesize{}}}{\footnotesize \par}

{\footnotesize{}{}}}{\footnotesize \par}

{\footnotesize{}{}{} }{\scriptsize{}{}  }
\end{table}
{\scriptsize \par}

{\scriptsize{}{}}{\scriptsize \par}

{\footnotesize{}{}{{}}}{\footnotesize \par}

{\footnotesize{}{}{{}}} 
\begin{table}[H]
\begin{centering}
{\footnotesize{}{}{{}}}%
\begin{tabular}{|c|c|cc|cc|cc|cc|}
\cline{2-10} 
\multicolumn{1}{c|}{} & {\footnotesize{}{}{{Analytic }}}  & \multicolumn{2}{c|}{{\footnotesize{}{}{{Numerical}}}} & \multicolumn{2}{c|}{{\footnotesize{}{}{{Quantization}}}} & \multicolumn{2}{c|}{{\footnotesize{}{}{{Monte Carlo}}}} & \multicolumn{2}{c|}{{\footnotesize{}{}{{MC-Sobol}}}}\tabularnewline
\hline 
{\footnotesize{}{}{{t }}}  & {\footnotesize{}{}{{EE }}}  & \multicolumn{1}{c}{{\footnotesize{}{}{{EE}}}} & {\footnotesize{}{}{{$\varepsilon$ }}}  & \multicolumn{1}{c}{{\footnotesize{}{}{{EE}}}} & {\footnotesize{}{}{{$\varepsilon$ }}}  & \multicolumn{1}{c}{{\footnotesize{}{}{{EE}}}} & {\footnotesize{}{}{{RSD }}}  & \multicolumn{1}{c}{{\footnotesize{}{}{{EE}}}} & {\footnotesize{}{}{{$\varepsilon$}}}\tabularnewline
\hline 
1w  & 18,0448  & 18,0435  & -0,007\%  & 18,0447  & 0,000\%  & 17,9530  & 0,495\%  & 18,0427  & -0,012\%\tabularnewline
2w  & 18,0551  & 18,0538  & -0,008\%  & 18,0552  & 0,000\%  & 18,0637  & 0,693\%  & 18,0524  & -0,015\%\tabularnewline
3w  & 18,0656  & 18,0640  & -0,009\%  & 18,0656  & 0,000\%  & 18,2050  & 0,880\%  & 18,0628  & -0,015\%\tabularnewline
1m  & 18,0760  & 18,0743  & -0,009\%  & 18,0760  & 0,000\%  & 18,2792  & 0,996\%  & 18,0754  & -0,004\%\tabularnewline
2m  & 18,1248  & 18,1225  & -0,013\%  & 18,1247  & 0,000\%  & 18,4457  & 1,420\%  & 18,1234  & -0,007\%\tabularnewline
3m  & 18,1701  & 18,1673  & -0,015\%  & 18,1701  & 0,000\%  & 18,1189  & 1,764\%  & 18,1627  & -0,041\%\tabularnewline
6m  & 18,3069  & 18,3027  & -0,023\%  & 18,3069  & 0,000\%  & 18,2130  & 2,568\%  & 18,3091  & 0,012\%\tabularnewline
9m  & 18,4447  & 18,4392  & -0,030\%  & 18,4447  & 0,000\%  & 18,6685  & 3,359\%  & 18,3983  & -0,252\%\tabularnewline
1y  & 18,5836  & 18,5763  & -0,036\%  & 18,5836  & 0,000\%  & 18,2320  & 3,985\%  & 18,5664  & -0,090\%\tabularnewline
\hline 
EPE  & 18,3638  & 18,3590  & -0,025\%  & 18,3638  & 0,000\%  & 18,33791  & -0,140\%  & 18,34755  & -0,088\%\tabularnewline
\hline 
\end{tabular}
\par\end{centering}

{\footnotesize{}{}{{\protect\protect\caption{{\footnotesize{}{}{{EPE: $10\%-$ITM European call. $\sigma=25\%$.\label{tab:ITM25}}}}}
}{\footnotesize \par}

{\footnotesize{}}}{\footnotesize \par}

{\footnotesize{}{}}}{\footnotesize \par}

{\footnotesize{}{}{}} 
\end{table}

{\footnotesize{}{}{{}}} 
\begin{table}[H]
\begin{centering}
{\footnotesize{}{}{{}}}%
\begin{tabular}{|c|c|cc|cc|cc|cc|}
\cline{2-10} 
\multicolumn{1}{c|}{} & {\footnotesize{}{}{{Analytic }}}  & \multicolumn{2}{c|}{{\footnotesize{}{}{{Numerical}}}} & \multicolumn{2}{c|}{{\footnotesize{}{}{{Quantization}}}} & \multicolumn{2}{c|}{{\footnotesize{}{}{{Monte Carlo}}}} & \multicolumn{2}{c|}{{\footnotesize{}{}{{MC-Sobol}}}}\tabularnewline
\hline 
{\footnotesize{}{}{{t }}}  & {\footnotesize{}{}{{EE }}}  & \multicolumn{1}{c}{{\footnotesize{}{}{{EE}}}} & {\footnotesize{}{}{{$\varepsilon$ }}}  & \multicolumn{1}{c}{{\footnotesize{}{}{{EE}}}} & {\footnotesize{}{}{{$\varepsilon$ }}}  & \multicolumn{1}{c}{{\footnotesize{}{}{{EE}}}} & {\footnotesize{}{}{{RSD }}}  & \multicolumn{1}{c}{{\footnotesize{}{}{{EE}}}} & {\footnotesize{}{}{{$\varepsilon$}}}\tabularnewline
\hline 
1w  & 19,884  & 19,883  & -0,007\%  & 19,884  & 0,000\%  & 19,777  & 0,525\%  & 19,882  & -0,012\%\tabularnewline
2w  & 19,896  & 19,894  & -0,008\%  & 19,895  & 0,000\%  & 19,905  & 0,735\%  & 19,892  & -0,016\%\tabularnewline
3w  & 19,907  & 19,905  & -0,009\%  & 19,907  & 0,000\%  & 20,073  & 0,935\%  & 19,904  & -0,016\%\tabularnewline
1m  & 19,918  & 19,917  & -0,010\%  & 19,918  & 0,000\%  & 20,157  & 1,058\%  & 19,918  & -0,004\%\tabularnewline
2m  & 19,972  & 19,969  & -0,014\%  & 19,972  & 0,000\%  & 20,338  & 1,510\%  & 19,971  & -0,008\%\tabularnewline
3m  & 20,022  & 20,019  & -0,017\%  & 20,022  & 0,000\%  & 19,947  & 1,881\%  & 20,013  & -0,045\%\tabularnewline
6m  & 20,173  & 20,168  & -0,026\%  & 20,173  & 0,000\%  & 20,028  & 2,758\%  & 20,177  & 0,018\%\tabularnewline
9m  & 20,325  & 20,318  & -0,035\%  & 20,325  & 0,000\%  & 20,608  & 3,622\%  & 20,278  & -0,230\%\tabularnewline
1y  & 20,478  & 20,468  & -0,044\%  & 20,478  & 0,000\%  & 20,083  & 4,309\%  & 20,509  & 0,156\%\tabularnewline
\hline 
EPE  & 20,236  & 20,229  & -0,030\%  & 20,236  & 0,000\%  & 20,204  & -0,155\%  & 20,232  & -0,019\%\tabularnewline
\hline 
\end{tabular}
\par\end{centering}

{\footnotesize{}{}{{\protect\protect\caption{{\footnotesize{}{}{{EPE: $10\%-$ITM European call. $\sigma=30\%$.\label{tab:ITM30}}}}}
}{\footnotesize \par}

{\footnotesize{}}}{\footnotesize \par}

{\footnotesize{}{}}}{\footnotesize \par}

{\footnotesize{}{}{}} 
\end{table}

{\footnotesize{}{}{{}}} 
\begin{table}[H]
\begin{centering}
{\footnotesize{}{}{{}}}%
\begin{tabular}{|c|c|cc|cc|cc|cc|}
\cline{2-10} 
\multicolumn{1}{c|}{} & {\footnotesize{}{}{{Analytic }}}  & \multicolumn{2}{c|}{{\footnotesize{}{}{{Numerical}}}} & \multicolumn{2}{c|}{{\footnotesize{}{}{{Quantization}}}} & \multicolumn{2}{c|}{{\footnotesize{}{}{{Monte Carlo}}}} & \multicolumn{2}{c|}{{\footnotesize{}{}{{MC-Sobol}}}}\tabularnewline
\hline 
{\footnotesize{}{}{{t }}}  & {\footnotesize{}{}{{EE }}}  & \multicolumn{1}{c}{{\footnotesize{}{}{{EE}}}} & {\footnotesize{}{}{{$\varepsilon$ }}}  & \multicolumn{1}{c}{{\footnotesize{}{}{{EE}}}} & {\footnotesize{}{}{{$\varepsilon$ }}}  & \multicolumn{1}{c}{{\footnotesize{}{}{{EE}}}} & {\footnotesize{}{}{{RSD }}}  & \multicolumn{1}{c}{{\footnotesize{}{}{{EE}}}} & {\footnotesize{}{}{{$\varepsilon$}}}\tabularnewline
\hline 
1w  & 7,48940  & 7,4889  & -0,007\%  & 7,4894  & 0,000\%  & 7,4479  & 0,539\%  & 7,4885  & -0,012\%\tabularnewline
2w  & 7,49372  & 7,4931  & -0,008\%  & 7,4937  & 0,000\%  & 7,4974  & 0,755\%  & 7,4925  & -0,016\%\tabularnewline
3w  & 7,49804  & 7,4974  & -0,009\%  & 7,4981  & 0,000\%  & 7,5623  & 0,960\%  & 7,4968  & -0,016\%\tabularnewline
1m  & 7,50237  & 7,5016  & -0,010\%  & 7,5024  & 0,000\%  & 7,5947  & 1,086\%  & 7,5021  & -0,004\%\tabularnewline
2m  & 7,52260  & 7,5216  & -0,013\%  & 7,5226  & 0,000\%  & 7,6646  & 1,552\%  & 7,5219  & -0,009\%\tabularnewline
3m  & 7,54143  & 7,5402  & -0,017\%  & 7,5414  & 0,000\%  & 7,5128  & 1,935\%  & 7,5379  & -0,046\%\tabularnewline
6m  & 7,59820  & 7,5964  & -0,024\%  & 7,5982  & 0,000\%  & 7,5412  & 2,842\%  & 7,6011  & 0,039\%\tabularnewline
9m  & 7,65540  & 7,6530  & -0,031\%  & 7,6554  & 0,000\%  & 7,7662  & 3,735\%  & 7,6354  & -0,262\%\tabularnewline
1y  & 7,7130  & 7,7099  & -0,037\%  & 7,7130  & 0,000\%  & 7,5699  & 4,473\%  & 7,7345  & 0,282\%\tabularnewline
\hline 
EPE  & 7,6218  & 7,61975  & -0,026\%  & 7,6218  & 0,000\%  & 7,61211  & -0,127\%  & 7,62250  & 0,010\%\tabularnewline
\hline 
\end{tabular}
\par\end{centering}

{\footnotesize{}{}{{\protect\protect\caption{{\footnotesize{}{}{{EPE: ATM European call. $\sigma=15\%$.\label{tab:ATM15}}}}}
}{\footnotesize \par}

{\footnotesize{}}}{\footnotesize \par}

{\footnotesize{}{}}}{\footnotesize \par}

{\footnotesize{}{}{}} 
\end{table}

{\footnotesize{}{}{{}}} 
\begin{table}[H]
\begin{centering}
{\footnotesize{}{}{{}}}%
\begin{tabular}{|c|c|cc|cc|cc|cc|}
\cline{2-10} 
\multicolumn{1}{c|}{} & {\footnotesize{}{}{{Analytic }}}  & \multicolumn{2}{c|}{{\footnotesize{}{}{{Numerical}}}} & \multicolumn{2}{c|}{{\footnotesize{}{}{{Quantization}}}} & \multicolumn{2}{c|}{{\footnotesize{}{}{{Monte Carlo}}}} & \multicolumn{2}{c|}{{\footnotesize{}{}{{MC-Sobol}}}}\tabularnewline
\hline 
{\footnotesize{}{}{{t }}}  & {\footnotesize{}{}{{EE }}}  & \multicolumn{1}{c}{{\footnotesize{}{}{{EE}}}} & {\footnotesize{}{}{{$\varepsilon$ }}}  & \multicolumn{1}{c}{{\footnotesize{}{}{{EE}}}} & {\footnotesize{}{}{{$\varepsilon$ }}}  & \multicolumn{1}{c}{{\footnotesize{}{}{{EE}}}} & {\footnotesize{}{}{{RSD }}}  & \multicolumn{1}{c}{{\footnotesize{}{}{{EE}}}} & {\footnotesize{}{}{{$\varepsilon$}}}\tabularnewline
\hline 
1w  & 11,3550  & 11,3542  & -0,007\%  & 11,3550  & 0,000\%  & 11,2873  & 0,583\%  & 11,3536  & -0,013\%\tabularnewline
2w  & 11,3616  & 11,3606  & -0,009\%  & 11,3616  & 0,000\%  & 11,3670  & 0,815\%  & 11,3597  & -0,016\%\tabularnewline
3w  & 11,3681  & 11,3670  & -0,010\%  & 11,3681  & 0,000\%  & 11,4762  & 1,039\%  & 11,3661  & -0,018\%\tabularnewline
1m  & 11,3747  & 11,3735  & -0,011\%  & 11,3747  & 0,000\%  & 11,5270  & 1,176\%  & 11,3742  & -0,004\%\tabularnewline
2m  & 11,4054  & 11,4036  & -0,016\%  & 11,4054  & 0,000\%  & 11,6289  & 1,684\%  & 11,4042  & -0,010\%\tabularnewline
3m  & 11,4339  & 11,4317  & -0,020\%  & 11,4339  & 0,000\%  & 11,3723  & 2,107\%  & 11,4279  & -0,053\%\tabularnewline
6m  & 11,5200  & 11,5165  & -0,030\%  & 11,5200  & 0,000\%  & 11,3929  & 3,132\%  & 11,5262  & 0,054\%\tabularnewline
9m  & 11,6067  & 11,6020  & -0,040\%  & 11,6067  & 0,000\%  & 11,8133  & 4,142\%  & 11,5793  & -0,236\%\tabularnewline
1y  & 11,6941  & 11,6879  & -0,050\%  & 11,6941  & 0,000\%  & 11,4752  & 4,978\%  & 11,7172  & 0,201\%\tabularnewline
\hline 
EPE  & 11,5558  & 11,5518  & -0,034\%  & 11,5558  & 0,000\%  & 11,53970  & -0,139\%  & 11,55554  & -0,001\%\tabularnewline
\hline 
\end{tabular}
\par\end{centering}

{\footnotesize{}{}{{\protect\protect\caption{{\footnotesize{}{}{{EPE: ATM European call. $\sigma=25\%$.\label{tab:ATM25}}}}}
}{\footnotesize \par}

{\footnotesize{}}}{\footnotesize \par}

{\footnotesize{}{}}}{\footnotesize \par}

{\footnotesize{}{}{}} 
\end{table}

{\footnotesize{}{}{{}}} 
\begin{table}[H]
\begin{centering}
{\footnotesize{}{}{{}}}%
\begin{tabular}{|c|c|cc|cc|cc|cc|}
\cline{2-10} 
\multicolumn{1}{c|}{} & {\footnotesize{}{}{{Analytic }}}  & \multicolumn{2}{c|}{{\footnotesize{}{}{{Numerical}}}} & \multicolumn{2}{c|}{{\footnotesize{}{}{{Quantization}}}} & \multicolumn{2}{c|}{{\footnotesize{}{}{{Monte Carlo}}}} & \multicolumn{2}{c|}{{\footnotesize{}{}{{MC-Sobol}}}}\tabularnewline
\hline 
{\footnotesize{}{}{{t }}}  & {\footnotesize{}{}{{EE }}}  & \multicolumn{1}{c}{{\footnotesize{}{}{{EE}}}} & {\footnotesize{}{}{{$\varepsilon$ }}}  & \multicolumn{1}{c}{{\footnotesize{}{}{{EE}}}} & {\footnotesize{}{}{{$\varepsilon$ }}}  & \multicolumn{1}{c}{{\footnotesize{}{}{{EE}}}} & {\footnotesize{}{}{{RSD }}}  & \multicolumn{1}{c}{{\footnotesize{}{}{{EE}}}} & {\footnotesize{}{}{{$\varepsilon$}}}\tabularnewline
\hline 
1w  & 13,291  & 13,290  & -0,008\%  & 13,291  & 0,000\%  & 13,209  & 0,600\%  & 13,289  & -0,013\%\tabularnewline
2w  & 13,298  & 13,297  & -0,009\%  & 13,298  & 0,000\%  & 13,305  & 0,839\%  & 13,296  & -0,016\%\tabularnewline
3w  & 13,306  & 13,305  & -0,010\%  & 13,306  & 0,000\%  & 13,438  & 1,070\%  & 13,303  & -0,019\%\tabularnewline
1m  & 13,314  & 13,312  & -0,011\%  & 13,314  & 0,000\%  & 13,498  & 1,211\%  & 13,313  & -0,005\%\tabularnewline
2m  & 13,349  & 13,347  & -0,016\%  & 13,349  & 0,000\%  & 13,614  & 1,736\%  & 13,348  & -0,010\%\tabularnewline
3m  & 13,383  & 13,380  & -0,021\%  & 13,383  & 0,000\%  & 13,301  & 2,176\%  & 13,375  & -0,056\%\tabularnewline
6m  & 13,484  & 13,479  & -0,034\%  & 13,484  & 0,000\%  & 13,314  & 3,251\%  & 13,491  & 0,057\%\tabularnewline
9m  & 13,585  & 13,579  & -0,045\%  & 13,585  & 0,000\%  & 13,846  & 4,312\%  & 13,555  & -0,220\%\tabularnewline
1y  & 13,687  & 13,679  & -0,057\%  & 13,687  & 0,000\%  & 13,422  & 5,193\%  & 13,711  & 0,172\%\tabularnewline
\hline 
EPE  & 13,526  & 13,520  & -0,038\%  & 13,526  & 0,000\%  & 13,504  & -0,161\%  & 13,525  & -0,004\%\tabularnewline
\hline 
\end{tabular}
\par\end{centering}

{\footnotesize{}{}{{\protect\protect\caption{{\footnotesize{}{}{{EPE: ATM European call. $\sigma=30\%$.\label{tab:ATM30}}}}}
}{\footnotesize \par}

{\footnotesize{}}}{\footnotesize \par}

{\footnotesize{}{}}}{\footnotesize \par}

{\footnotesize{}{}{}} 
\end{table}

{\footnotesize{}{}{{}}} 
\begin{table}[H]
\begin{centering}
{\footnotesize{}{}{{}}}%
\begin{tabular}{|c|c|cc|cc|cc|cc|}
\cline{2-10} 
\multicolumn{1}{c|}{} & {\footnotesize{}{}{{Analytic }}}  & \multicolumn{2}{c|}{{\footnotesize{}{}{{Numerical}}}} & \multicolumn{2}{c|}{{\footnotesize{}{}{{Quantization}}}} & \multicolumn{2}{c|}{{\footnotesize{}{}{{Monte Carlo}}}} & \multicolumn{2}{c|}{{\footnotesize{}{}{{MC-Sobol}}}}\tabularnewline
\hline 
{\footnotesize{}{}{{t }}}  & {\footnotesize{}{}{{EE }}}  & \multicolumn{1}{c}{{\footnotesize{}{}{{EE}}}} & {\footnotesize{}{}{{$\varepsilon$ }}}  & \multicolumn{1}{c}{{\footnotesize{}{}{{EE}}}} & {\footnotesize{}{}{{$\varepsilon$ }}}  & \multicolumn{1}{c}{{\footnotesize{}{}{{EE}}}} & {\footnotesize{}{}{{RSD }}}  & \multicolumn{1}{c}{{\footnotesize{}{}{{EE}}}} & {\footnotesize{}{}{{$\varepsilon$}}}\tabularnewline
\hline 
1w  & 2,7600  & 2,7598  & -0,008\%  & 2,7600  & 0,000\%  & 2,7396  & 0,731\%  & 2,7597  & -0,012\%\tabularnewline
2w  & 2,7616  & 2,76135  & -0,010\%  & 2,7616  & 0,000\%  & 2,7628  & 1,023\%  & 2,7612  & -0,016\%\tabularnewline
3w  & 2,7632  & 2,7629  & -0,012\%  & 2,7632  & 0,000\%  & 2,7987  & 1,310\%  & 2,7626  & -0,016\%\tabularnewline
1m  & 2,7649  & 2,7644  & -0,014\%  & 2,7649  & 0,000\%  & 2,8121  & 1,483\%  & 2,7647  & -0,004\%\tabularnewline
2m  & 2,7723  & 2,7716  & -0,023\%  & 2,7723  & 0,000\%  & 2,8325  & 2,145\%  & 2,7720  & -0,009\%\tabularnewline
3m  & 2,7792  & 2,7784  & -0,030\%  & 2,7792  & 0,000\%  & 2,7476  & 2,725\%  & 2,7772  & -0,046\%\tabularnewline
6m  & 2,8001  & 2,7988  & -0,049\%  & 2,8001  & 0,000\%  & 2,7319  & 4,239\%  & 2,8055  & 0,039\%\tabularnewline
9m  & 2,8212  & 2,8194  & -0,065\%  & 2,8212  & 0,000\%  & 2,9162  & 5,719\%  & 2,8119  & -0,262\%\tabularnewline
1y  & 2,8424  & 2,8401  & -0,080\%  & 2,8424  & 0,000\%  & 2,8243  & 7,036\%  & 2,8258  & 0,282\%\tabularnewline
\hline 
EPE  & 2,8088  & 2,80730  & -0,054\%  & 2,8088  & 0,000\%  & 2,81496  & 0,218\%  & 2,80347  & -0,191\%\tabularnewline
\hline 
\end{tabular}
\par\end{centering}

{\footnotesize{}{}{{\protect\protect\caption{{\footnotesize{}{}{{EPE: $10\%-$OTM European call. $\sigma=15\%$.\label{tab:OTM15}}}}}
}{\footnotesize \par}

{\footnotesize{}}}{\footnotesize \par}

{\footnotesize{}{}}}{\footnotesize \par}

{\footnotesize{}{}{}} 
\end{table}

{\footnotesize{}{}{{}}} 
\begin{table}[H]
\begin{centering}
{\footnotesize{}{}{{}}}%
\begin{tabular}{|c|c|cc|cc|cc|cc|}
\cline{2-10} 
\multicolumn{1}{c|}{} & {\footnotesize{}{}{{Analytic }}}  & \multicolumn{2}{c|}{{\footnotesize{}{}{{Numerical}}}} & \multicolumn{2}{c|}{{\footnotesize{}{}{{Quantization}}}} & \multicolumn{2}{c|}{{\footnotesize{}{}{{Monte Carlo}}}} & \multicolumn{2}{c|}{{\footnotesize{}{}{{MC-Sobol}}}}\tabularnewline
\hline 
{\footnotesize{}{}{{t }}}  & {\footnotesize{}{}{{EE }}}  & \multicolumn{1}{c}{{\footnotesize{}{}{{EE}}}} & {\footnotesize{}{}{{$\varepsilon$ }}}  & \multicolumn{1}{c}{{\footnotesize{}{}{{EE}}}} & {\footnotesize{}{}{{$\varepsilon$ }}}  & \multicolumn{1}{c}{{\footnotesize{}{}{{EE}}}} & {\footnotesize{}{}{{RSD }}}  & \multicolumn{1}{c}{{\footnotesize{}{}{{EE}}}} & {\footnotesize{}{}{{$\varepsilon$}}}\tabularnewline
\hline 
1w  & 6,2016  & 6,2011  & -0,008\%  & 6,2016  & 0,000\%  & 6,1579  & 0,695\%  & 6,2008  & -0,013\%\tabularnewline
2w  & 6,2052  & 6,2046  & -0,010\%  & 6,2052  & 0,000\%  & 6,2080  & 0,973\%  & 6,2042  & -0,016\%\tabularnewline
3w  & 6,2088  & 6,2081  & -0,012\%  & 6,2088  & 0,000\%  & 6,2835  & 1,245\%  & 6,2074  & -0,018\%\tabularnewline
1m  & 6,2124  & 6,2116  & -0,013\%  & 6,2124  & 0,000\%  & 6,3131  & 1,409\%  & 6,2121  & -0,004\%\tabularnewline
2m  & 6,2291  & 6,2278  & -0,021\%  & 6,2291  & 0,000\%  & 6,3616  & 2,033\%  & 6,2285  & -0,010\%\tabularnewline
3m  & 6,2447  & 6,2430  & -0,027\%  & 6,2447  & 0,000\%  & 6,1831  & 2,573\%  & 6,2405  & -0,053\%\tabularnewline
6m  & 6,2917  & 6,2889  & -0,045\%  & 6,2917  & 0,000\%  & 6,1594  & 3,956\%  & 6,3005  & 0,054\%\tabularnewline
9m  & 6,3391  & 6,3352  & -0,062\%  & 6,3391  & 0,000\%  & 6,5221  & 5,314\%  & 6,3217  & -0,236\%\tabularnewline
1y  & 6,3868  & 6,3817  & -0,077\%  & 6,3868  & 0,000\%  & 6,3051  & 6,501\%  & 6,3418  & 0,201\%\tabularnewline
\hline 
EPE  & 6,3113  & 6,3080  & -0,051\%  & 6,3113  & 0,000\%  & 6,31287  & 0,026\%  & 6,29739  & -0,220\%\tabularnewline
\hline 
\end{tabular}
\par\end{centering}

{\footnotesize{}{}{{\protect\protect\caption{{\footnotesize{}{}{{EPE: $10\%-$OTM European call. $\sigma=25\%$.\label{tab:OTM25}}}}}
}{\footnotesize \par}

{\footnotesize{}}}{\footnotesize \par}

{\footnotesize{}{}}}{\footnotesize \par}

{\footnotesize{}{}{}} 
\end{table}

{\footnotesize{}{}{{}}} 
\begin{table}[H]
\begin{centering}
{\footnotesize{}{}{{}}}%
\begin{tabular}{|c|c|cc|cc|cc|cc|}
\cline{2-10} 
\multicolumn{1}{c|}{} & {\footnotesize{}{}{{Analytic }}}  & \multicolumn{2}{c|}{{\footnotesize{}{}{{Numerical}}}} & \multicolumn{2}{c|}{{\footnotesize{}{}{{Quantization}}}} & \multicolumn{2}{c|}{{\footnotesize{}{}{{Monte Carlo}}}} & \multicolumn{2}{c|}{{\footnotesize{}{}{{MC-Sobol}}}}\tabularnewline
\hline 
{\footnotesize{}{}{{t }}}  & {\footnotesize{}{}{{EE }}}  & \multicolumn{1}{c}{{\footnotesize{}{}{{EE}}}} & {\footnotesize{}{}{{$\varepsilon$ }}}  & \multicolumn{1}{c}{{\footnotesize{}{}{{EE}}}} & {\footnotesize{}{}{{$\varepsilon$ }}}  & \multicolumn{1}{c}{{\footnotesize{}{}{{EE}}}} & {\footnotesize{}{}{{RSD }}}  & \multicolumn{1}{c}{{\footnotesize{}{}{{EE}}}} & {\footnotesize{}{}{{$\varepsilon$}}}\tabularnewline
\hline 
1w  & 7,9807  & 7,9800  & -0,008\%  & 7,9807  & 0,000\%  & 7,9245  & 0,693\%  & 7,9795  & -0,013\%\tabularnewline
2w  & 7,9853  & 7,9845  & -0,010\%  & 7,9853  & 0,000\%  & 7,9889  & 0,970\%  & 7,9839  & -0,016\%\tabularnewline
3w  & 7,9899  & 7,9889  & -0,011\%  & 7,9899  & 0,000\%  & 8,0857  & 1,241\%  & 7,9881  & -0,019\%\tabularnewline
1m  & 7,9945  & 7,9934  & -0,013\%  & 7,9945  & 0,000\%  & 8,1237  & 1,405\%  & 7,9941  & -0,005\%\tabularnewline
2m  & 8,0160  & 8,0144  & -0,021\%  & 8,0160  & 0,000\%  & 8,1858  & 2,027\%  & 8,0153  & -0,010\%\tabularnewline
3m  & 8,0361  & 8,0339  & -0,028\%  & 8,0361  & 0,000\%  & 7,9568  & 2,564\%  & 8,0306  & -0,056\%\tabularnewline
6m  & 8,0966  & 8,0928  & -0,046\%  & 8,0966  & 0,000\%  & 7,9268  & 3,942\%  & 8,1067  & 0,057\%\tabularnewline
9m  & 8,1575  & 8,1523  & -0,064\%  & 8,1575  & 0,000\%  & 8,3893  & 5,295\%  & 8,1376  & -0,220\%\tabularnewline
1y  & 8,2189  & 8,2120  & -0,082\%  & 8,2189  & 0,000\%  & 8,0964  & 6,475\%  & 8,1760  & 0,172\%\tabularnewline
\hline 
EPE  & 8,1218  & 8,11738  & -0,053\%  & 8,1218  & 0,000\%  & 8,11858  & -0,038\%  & 8,10792  & -0,170\%\tabularnewline
\hline 
\end{tabular}
\par\end{centering}

{\footnotesize{}{}{{\protect\protect\caption{{\footnotesize{}{}{{EPE: $10\%-$OTM European call. $\sigma=30\%$.\label{tab:OTM30}}}}}
}{\footnotesize \par}

{\footnotesize{}}}{\footnotesize \par}

{\footnotesize{}{}}}{\footnotesize \par}

{\footnotesize{}{}{}} 
\end{table}

Comparing EPE values reported in tables, we deduce that the quantization
approach provides highly satisfactory results for ATM, ITM, OTM European call options.

We also note that the Monte Carlo relative errors
increase for the out of the money cases. This is due to the fact
that in the considered framework, the true value of EE becomes very small.
It is worth to note that the VQ also works better than the numerical
integration.

For the sake of completeness and in order to stress how the quantization
technique perform better than Monte Carlo method, we report two figures showing the error $\varepsilon$ for Monte Carlo and quantization
performances. The plots were obtained for a more granular combination of the couple
(Spot,Volatility) values. 

\begin{figure}[H]
\begin{centering}
\includegraphics[scale=0.5]{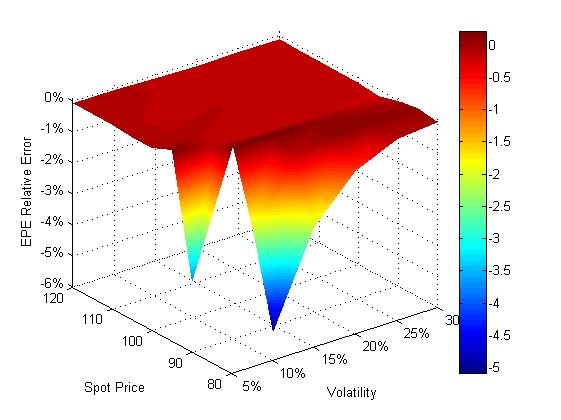} 
\par\end{centering}

{\footnotesize{{\caption{{\footnotesize{{MC: EPE error $\varepsilon$ for European Call.\label{fig:OTM_Call_eu_MC}}}}}
}}{\footnotesize \par}

{\footnotesize{}}} 
\end{figure}

\begin{figure}[H]
\begin{centering}
\includegraphics[scale=0.5]{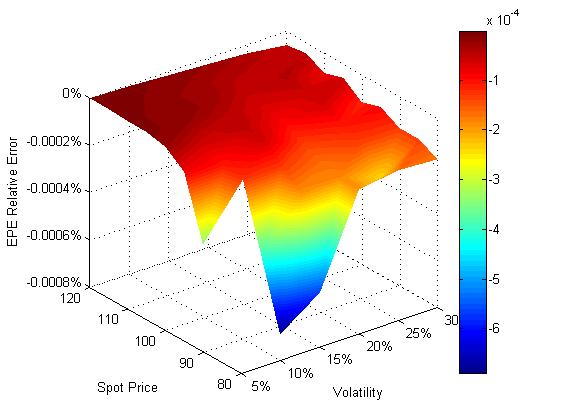} 
\par\end{centering}

{\footnotesize{{\caption{{\footnotesize{{VQ: error $\varepsilon$ for European Call.\label{fig:OTM_Call_eu_VQ}}}}}
}}{\footnotesize \par}

{\footnotesize{}}} 
\end{figure}

\subsection{A portfolio of European options}

\label{num3}

In order to generalize techniques and results shown in the previous
subsection, we are going to consider now a set of $n$ European options,
i.e. a derivative portfolio, for which we will evaluate the Expected
Exposure (EE) and the Expected Positive Exposure (EPE).

The portfolio may consist of call options and put options, related
to a group of trans\-actions with a single counterparty, which are
subject to a legally enforceable bilateral netting arrangement. Such
a set of derivatives is called \emph{netting set.} From a mathematical
point of view, such a choice means that the expected exposure is given
by 
\[
EE=\left(\sum_{j=1}^{n}MtM(t)-V_{t}\right)^{+}\;,
\]
unlike the case of no-netting portfolio setting, where the sum of
the fair prices of European options represents the Mark-to-Market
of the portfolio. The term $V_{t}$ represent the \emph{collateral}
value posted by the debtor, i.e. the counterparty with the negative
Mark-to-Market portfolio at time \emph{t}. In what follows we
set $V_{t}=0$, therefore we deal with the netting agreement situation,
supposing  no collateral. Even if the set up can appear simple, this is not true, and the
problem turns out to be rather challenging.
That is because, generally speaking, one can
not perform the analytical calculation of $EE_{k}$. In fact, also
if the martingale property holds for each derivative in the portfolio,
in the present case the positive part operator is effective, hence the future
value is not simply the compounded current MtM.

For computational purposes, we consider $n=10$ European options,
that is $5$ call options, the first, the third and the last one are
of \emph{buy} type, while the second and the fourth one are of \emph{sell}
type, and $5$ put options, namely, the first, the second and the
last one of \emph{sell} type and the remaining of \emph{buy} type.
To summarize, we have constructed a table with the features of the
different options, see Table \ref{table_ptf}.

The software code is quite general to deal with any change in quantities,
buy/sell, market and instrument parameters.

For the application we consider the following values: 
\begin{itemize}
\item spot price $(S_{0}):$ $90,100,110$ 
\item interest rate $(r):$ $3\%$ 
\item volatility $(\sigma):$ $15\%,25\%,30\%$ 
\item maturity $(T):$ one year 
\item time buckets: $\left\{ 0,\frac{1}{52},\frac{2}{52},\frac{3}{52},\frac{4}{52},\frac{2}{12},\frac{3}{12},\frac{6}{12},\frac{9}{12},1\right\} .$ 
\end{itemize}
\begin{table}[H]
\begin{centering}
{\footnotesize{}{}{}}%
\begin{tabular}{|c|c|c|c|c|}
\hline 
 & type  & position  & strike  & maturity \tabularnewline
\hline 
Option 1  & call  & buy  & $125$  & 1 year \tabularnewline
Option 2  & call  & sell  & $100$  & 1 year \tabularnewline
Option 3  & call  & buy  & $80$  & 1 year \tabularnewline
Option 4  & call  & sell  & $95$  & 1 year \tabularnewline
Option 5  & call  & buy  & $105$  & 1 year \tabularnewline
Option 6  & put  & sell  & $80$  & 1 year \tabularnewline
Option 7  & put  & sell  & $100$  & 1 year \tabularnewline
Option 8  & put  & buy  & $110$  & 1 year \tabularnewline
Option 9  & put  & buy  & $90$  & 1 year \tabularnewline
Option 10  & put  & sell  & $120$  & 1 year \tabularnewline
\hline 
\end{tabular}
\par\end{centering}

{\footnotesize{}{}{\protect\protect\caption{{\footnotesize{}{}{Description of the portfolio. The table summarizes
the characteristics of the netting set}}}
}{\footnotesize \par}

{\footnotesize{}}}{\footnotesize \par}

{\footnotesize{}{}} \label{table_ptf} 
\end{table}

To better understand the characteristics of the portfolio, let us
consider the following two figures, which show the portfolio profile
for the different volatilities. It is a \emph{long} style portfolio,
with different levels of \emph{delta} sensitivities.

\begin{figure}[H]
\begin{centering}
\includegraphics[scale=0.5]{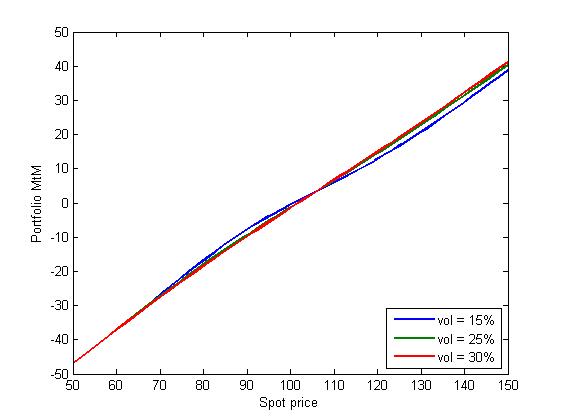} 
\par\end{centering}

{\footnotesize{}{}{{\protect\protect\caption{{\footnotesize{}{}{{MtM of the portfolio depending on volatility.
\label{fig:MtM}}}}}
}{\footnotesize \par}

{\footnotesize{}}}{\footnotesize \par}

{\footnotesize{}{}}}{\footnotesize \par}

{\footnotesize{}{}{}} 
\end{figure}

\begin{figure}[H]
\begin{centering}
\includegraphics[scale=0.5]{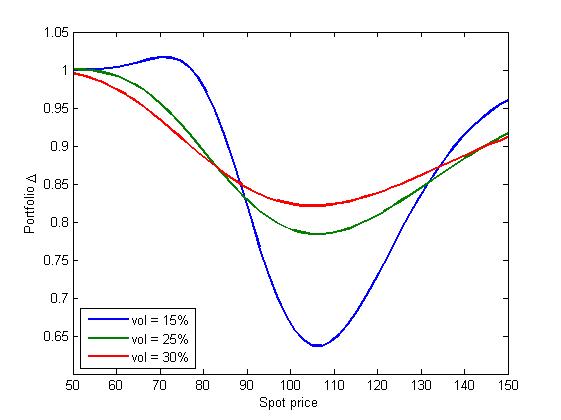} 
\par\end{centering}

{\footnotesize{}{}{{\protect\protect\caption{{\footnotesize{}{}{{Delta of the portfolio depending on volatility.
\label{fig:delta}}}}}
}{\footnotesize \par}

{\footnotesize{}}}{\footnotesize \par}

{\footnotesize{}{}}}{\footnotesize \par}

{\footnotesize{}{}{}} 
\end{figure}

The benchmark value for this application can not be a closed Black-Scholes
approach. Hence we use a Monte Carlo-Sobol sequence with a
large enough number of points as an acceptable pivot value. We use
$10^{6}$ points. The vector quantization and the Monte Carlo techniques
are tested with $10^{3}$ points.

As it was done in Subsection \ref{num2}, in relation to the case
of only one option, the comparison among the different procedures
consists in analyzing the percent relative standard error (RSD) for
the Monte Carlo approach and the percent relative error ($\varepsilon$)
for the quantization technique.

\begin{table}[H]
\begin{centering}
{\footnotesize{}{}{}}%
\begin{tabular}{|c|c|cc|cc|cc|}
\cline{2-8} 
\multicolumn{1}{c|}{} & {\footnotesize{}{}{MC-Sobol $10^{6}$ }}  & \multicolumn{2}{c|}{{\footnotesize{}{}{Quantization}}} & \multicolumn{2}{c|}{{\footnotesize{}{}{Monte Carlo}}} & \multicolumn{2}{c|}{{\footnotesize{}{}{MC-Sobol $10^{3}$}}}\tabularnewline
\hline 
{\footnotesize{}{}{t }}  & {\footnotesize{}{}{EE }}  & \multicolumn{1}{c}{{\footnotesize{}{}{EE}}} & {\footnotesize{}{}{$\varepsilon$ }}  & \multicolumn{1}{c}{{\footnotesize{}{}{EE}}} & {\footnotesize{}{}{RSD }}  & \multicolumn{1}{c}{{\footnotesize{}{}{EE}}} & {\footnotesize{}{}{RSD }} \tabularnewline
\hline 
1w  & 0,0000  & 0,0000  & NaN  & 0,0000  & NaN  & 0,0000  & NaN \tabularnewline
2w  & 0,0000  & 0,0000  & 4,184\%  & 0,0000  & NaN  & 0,0002  & 99,950\% \tabularnewline
3w  & 0,0006  & 0,0006  & -0,585\%  & 0,0007  & 99,950\%  & 0,0004  & 99,950\% \tabularnewline
1m  & 0,0030  & 0,0030  & -0,021\%  & 0,0022  & 70,933\%  & 0,0060  & 58,952\% \tabularnewline
2m  & 0,0537  & 0,0537  & 0,033\%  & 0,0425  & 22,046\%  & 0,0586  & 20,217\% \tabularnewline
3m  & 0,1462  & 0,1463  & 0,048\%  & 0,1254  & 15,505\%  & 0,1552  & 15,306\% \tabularnewline
6m  & 0,5045  & 0,5045  & 0,010\%  & 0,4539  & 10,512\%  & 0,5054  & 10,581\% \tabularnewline
9m  & 0,8529  & 0,8529  & -0,005\%  & 0,8965  & 8,428\%  & 0,8926  & 8,581\% \tabularnewline
1y  & 1,3863  & 1,3874  & 0,078\%  & 1,2717  & 7,466\%  & 1,4630  & 7,294\% \tabularnewline
\hline 
EPE  & 0,7030  & 0,7033  & 0,040\%  & 0,6698  & -4,719\%  & 0,7336  & 4,348\% \tabularnewline
\hline 
\end{tabular}
\par\end{centering}

{\footnotesize{}{}{\protect\protect\caption{{\footnotesize{}{}{EPE: $S_{0}=90;\sigma=15\%.$ \label{tab:90.15}}}}
}{\footnotesize \par}

{\footnotesize{}}}{\footnotesize \par}

{\footnotesize{}{}} 
\end{table}

\begin{table}[H]
\begin{centering}
{\footnotesize{}{}{}}%
\begin{tabular}{|c|c|cc|cc|cc|}
\cline{2-8} 
\multicolumn{1}{c|}{} & {\footnotesize{}{}{MC-Sobol $10^{6}$ }}  & \multicolumn{2}{c|}{{\footnotesize{}{}{Quantization}}} & \multicolumn{2}{c|}{{\footnotesize{}{}Monte Carlo}} & \multicolumn{2}{c|}{{\footnotesize{}{}{MC-Sobol $10^{3}$}}}\tabularnewline
\hline 
{\footnotesize{}{}{t }}  & {\footnotesize{}{}{EE }}  & \multicolumn{1}{c}{{\footnotesize{}{}{EE}}} & {\footnotesize{}{}{$\varepsilon$ }}  & \multicolumn{1}{c}{{\footnotesize{}{}{EE}}} & {\footnotesize{}{}{RSD }}  & \multicolumn{1}{c}{{\footnotesize{}{}{EE}}} & {\footnotesize{}{}{RSD }} \tabularnewline
\hline 
1w  & 0,0001  & 0,0001  & 0,701\%  & 0,0001  & 99,950\%  & 0,0000  & NaN \tabularnewline
2w  & 0,0075  & 0,0075  & 0,032\%  & 0,0074  & 54,701\%  & 0,0093  & 64,637\% \tabularnewline
3w  & 0,0348  & 0,0348  & -0,046\%  & 0,0523  & 24,373\%  & 0,0427  & 27,104\% \tabularnewline
1m  & 0,0817  & 0,0817  & 0,014\%  & 0,0939  & 18,901\%  & 0,0899  & 23,586\% \tabularnewline
2m  & 0,4319  & 0,4319  & 0,001\%  & 0,4329  & 11,475\%  & 0,4489  & 12,248\% \tabularnewline
3m  & 0,8119  & 0,8121  & 0,028\%  & 0,7440  & 10,323\%  & 0,8226  & 10,542\% \tabularnewline
6m  & 1,8836  & 1,8837  & 0,005\%  & 1,7520  & 8,568\%  & 1,9046  & 8,687\% \tabularnewline
9m  & 2,7611  & 2,7612  & 0,005\%  & 2,9002  & 7,896\%  & 2,8727  & 8,103\% \tabularnewline
1y  & 3,5191  & 3,5217  & 0,075\%  & 3,4064  & 7,868\%  & 3,5674  & 8,452\% \tabularnewline
\hline 
EPE  & 2,1497  & 2,1505  & 0,034\%  & 2,1185  & -1,454\%  & 2,1977  & 2,232\% \tabularnewline
\hline 
\end{tabular}
\par\end{centering}

{\footnotesize{}{}{\protect\protect\caption{{\footnotesize{}{}{EPE: $S_{0}=90,\sigma=25\%.$ \label{tab:90.25}}}}
}{\footnotesize \par}

{\footnotesize{}}}{\footnotesize \par}

{\footnotesize{}{}} 
\end{table}

\begin{table}[H]
\begin{centering}
{\footnotesize{}{}{}}%
\begin{tabular}{|c|c|cc|cc|cc|}
\cline{2-8} 
\multicolumn{1}{c|}{} & {\footnotesize{}{}{MC-Sobol $10^{6}$ }}  & \multicolumn{2}{c|}{{\footnotesize{}{}{Quantization}}} & \multicolumn{2}{c|}{{\footnotesize{}{}Monte Carlo}} & \multicolumn{2}{c|}{{\footnotesize{}{}{MC-Sobol $10^{3}$}}}\tabularnewline
\hline 
{\footnotesize{}{}{t }}  & {\footnotesize{}{}{EE }}  & \multicolumn{1}{c}{{\footnotesize{}{}{EE}}} & {\footnotesize{}{}{$\varepsilon$ }}  & \multicolumn{1}{c}{{\footnotesize{}{}{EE}}} & {\footnotesize{}{}{RSD }}  & \multicolumn{1}{c}{{\footnotesize{}{}{EE}}} & {\footnotesize{}{}{RSD }} \tabularnewline
\hline 
1w  & 0,0013  & 0,0013  & 0,145\%  & 0,0040  & 68,879\%  & 0,0005  & 99,950\% \tabularnewline
2w  & 0,0291  & 0,0291  & 0,004\%  & 0,0325  & 30,112\%  & 0,0312  & 37,205\% \tabularnewline
3w  & 0,0981  & 0,0981  & -0,013\%  & 0,1321  & 18,842\%  & 0,1157  & 19,728\% \tabularnewline
1m  & 0,1954  & 0,1954  & -0,012\%  & 0,2277  & 14,700\%  & 0,2049  & 17,692\% \tabularnewline
2m  & 0,7796  & 0,7796  & 0,004\%  & 0,7821  & 10,110\%  & 0,7997  & 10,680\% \tabularnewline
3m  & 1,3415  & 1,3418  & 0,021\%  & 1,2408  & 9,292\%  & 1,3594  & 9,441\% \tabularnewline
6m  & 2,8256  & 2,8257  & 0,005\%  & 2,6487  & 8,089\%  & 2,8590  & 8,221\% \tabularnewline
9m  & 4,0074  & 4,0077  & 0,008\%  & 4,2019  & 7,686\%  & 4,1538  & 7,896\% \tabularnewline
1y  & 4,9419  & 4,9447  & 0,058\%  & 4,8168  & 7,829\%  & 5,0104  & 8,410\% \tabularnewline
\hline 
EPE  & 3,1317  & 3,1325  & 0,027\%  & 3,0981  & -1,074\%  & 3,1976  & 2,106\% \tabularnewline
\hline 
\end{tabular}
\par\end{centering}

{\footnotesize{}{}{\protect\protect\caption{{\footnotesize{}{}{EPE: $S_{0}=90,\sigma=30\%.$ \label{tab:90.30}}}}
}{\footnotesize \par}

{\footnotesize{}}}{\footnotesize \par}

{\footnotesize{}{}} 
\end{table}

\begin{table}[H]
\begin{centering}
{\footnotesize{}{}{}}%
\begin{tabular}{|c|c|cc|cc|cc|}
\cline{2-8} 
\multicolumn{1}{c|}{} & {\footnotesize{}{}{MC-Sobol $10^{6}$ }}  & \multicolumn{2}{c|}{{\footnotesize{}{}{Quantization}}} & \multicolumn{2}{c|}{{\footnotesize{}{}Monte Carlo}} & \multicolumn{2}{c|}{{\footnotesize{}{}{MC sobol $10^{3}$}}}\tabularnewline
\hline 
{\footnotesize{}{}{t }}  & {\footnotesize{}{}{EE }}  & \multicolumn{1}{c}{{\footnotesize{}{}{EE}}} & {\footnotesize{}{}{$\varepsilon$ }}  & \multicolumn{1}{c}{{\footnotesize{}{}{EE}}} & {\footnotesize{}{}{RSD }}  & \multicolumn{1}{c}{{\footnotesize{}{}{EE}}} & {\footnotesize{}{}{RSD }} \tabularnewline
\hline 
1w  & 0,3565  & 0,3565  & 0,000\%  & 0,3325  & 6,223\%  & 0,3562  & 5,746\% \tabularnewline
2w  & 0,5758  & 0,5758  & 0,000\%  & 0,5681  & 5,447\%  & 0,5757  & 5,433\% \tabularnewline
3w  & 0,7463  & 0,7463  & 0,000\%  & 0,7992  & 5,182\%  & 0,7482  & 5,275\% \tabularnewline
1m  & 0,8904  & 0,8904  & -0,003\%  & 0,9439  & 5,022\%  & 0,8889  & 5,266\% \tabularnewline
2m  & 1,3966  & 1,3966  & 0,000\%  & 1,4407  & 4,793\%  & 1,4016  & 5,007\% \tabularnewline
3m  & 1,7463  & 1,7463  & 0,004\%  & 1,6845  & 4,950\%  & 1,7351  & 5,051\% \tabularnewline
6m  & 2,5014  & 2,5014  & 0,003\%  & 2,4279  & 4,960\%  & 2,5149  & 5,097\% \tabularnewline
9m  & 3,0139  & 3,0138  & -0,002\%  & 3,0990  & 5,225\%  & 3,0437  & 5,377\% \tabularnewline
1y  & 3,5845  & 3,5852  & 0,018\%  & 3,5369  & 5,253\%  & 3,6740  & 5,459\% \tabularnewline
\hline 
EPE  & 2,5952  & 2,5954  & 0,007\%  & 2,5865  & -0,337\%  & 2,6279  & 1,261\% \tabularnewline
\hline 
\end{tabular}
\par\end{centering}

{\footnotesize{}{}{\protect\protect\caption{{\footnotesize{}{}{EPE: $S_{0}=100,\sigma=15\%.$ \label{tab:100.15}}}}
}{\footnotesize \par}

{\footnotesize{}}}{\footnotesize \par}

{\footnotesize{}{}} 
\end{table}

\begin{table}[H]
\begin{centering}
{\footnotesize{}{}{}}%
\begin{tabular}{|c|c|cc|cc|cc|}
\cline{2-8} 
\multicolumn{1}{c|}{} & {\footnotesize{}{}{MC-Sobol $10^{6}$ }}  & \multicolumn{2}{c|}{{\footnotesize{}{}{Quantization}}} & \multicolumn{2}{c|}{{\footnotesize{}{}Monte Carlo}} & \multicolumn{2}{c|}{{\footnotesize{}{}{MC-Sobol $10^{3}$}}}\tabularnewline
\hline 
{\footnotesize{}{}{t }}  & {\footnotesize{}{}{EE }}  & \multicolumn{1}{c}{{\footnotesize{}{}{EE}}} & {\footnotesize{}{}{$\varepsilon$ }}  & \multicolumn{1}{c}{{\footnotesize{}{}{EE}}} & {\footnotesize{}{}{RSD }}  & \multicolumn{1}{c}{{\footnotesize{}{}{EE}}} & {\footnotesize{}{}{RSD }} \tabularnewline
\hline 
1w  & 0,5510  & 0,5510  & 0,000\%  & 0,5113  & 7,316\%  & 0,5500  & 6,632\% \tabularnewline
2w  & 0,9683  & 0,9683  & 0,000\%  & 0,9532  & 6,113\%  & 0,9677  & 6,094\% \tabularnewline
3w  & 1,2999  & 1,2999  & 0,000\%  & 1,4062  & 5,717\%  & 1,2986  & 5,884\% \tabularnewline
1m  & 1,5831  & 1,5831  & 0,001\%  & 1,6824  & 5,537\%  & 1,5901  & 5,777\% \tabularnewline
2m  & 2,5909  & 2,5909  & 0,001\%  & 2,6600  & 5,273\%  & 2,5990  & 5,501\% \tabularnewline
3m  & 3,2975  & 3,2977  & 0,004\%  & 3,1569  & 5,466\%  & 3,2792  & 5,561\% \tabularnewline
6m  & 4,8611  & 4,8614  & 0,006\%  & 4,6669  & 5,563\%  & 4,8787  & 5,676\% \tabularnewline
9m  & 5,9723  & 5,9725  & 0,002\%  & 6,1526  & 5,816\%  & 6,0292  & 6,002\% \tabularnewline
1y  & 6,8363  & 6,8377  & 0,020\%  & 6,6731  & 6,139\%  & 6,9781  & 6,320\% \tabularnewline
\hline 
EPE  & 5,0094  & 5,0099  & 0,009\%  & 4,9625  & -0,936\%  & 5,0627  & 1,065\% \tabularnewline
\hline 
\end{tabular}
\par\end{centering}

{\footnotesize{}{}{\protect\protect\caption{{\footnotesize{}{}{EPE: $S_{0}=100,\sigma=25\%.$ \label{tab:100.25}}}}
}{\footnotesize \par}

{\footnotesize{}}}{\footnotesize \par}

{\footnotesize{}{}} 
\end{table}

\begin{table}[H]
\begin{centering}
{\footnotesize{}{}{}}%
\begin{tabular}{|c|c|cc|cc|cc|}
\cline{2-8} 
\multicolumn{1}{c|}{} & {\footnotesize{}{}{MC-Sobol $10^{6}$ }}  & \multicolumn{2}{c|}{{\footnotesize{}{}{Quantization}}} & \multicolumn{2}{c|}{{\footnotesize{}{}Monte Carlo}} & \multicolumn{2}{c|}{{\footnotesize{}{}{MC-Sobol $10^{3}$}}}\tabularnewline
\hline 
{\footnotesize{}{}{t }}  & {\footnotesize{}{}{EE }}  & \multicolumn{1}{c}{{\footnotesize{}{}{EE}}} & {\footnotesize{}{}{$\varepsilon$ }}  & \multicolumn{1}{c}{{\footnotesize{}{}{EE}}} & {\footnotesize{}{}{RSD }}  & \multicolumn{1}{c}{{\footnotesize{}{}{EE}}} & {\footnotesize{}{}{RSD }} \tabularnewline
\hline 
1w  & 0,7348  & 0,7348  & 0,000\%  & 0,6826  & 7,124\%  & 0,7335  & 6,474\% \tabularnewline
2w  & 1,2651  & 1,2651  & 0,000\%  & 1,2462  & 6,034\%  & 1,2643  & 6,021\% \tabularnewline
3w  & 1,6844  & 1,6844  & -0,001\%  & 1,8210  & 5,684\%  & 1,6836  & 5,844\% \tabularnewline
1m  & 2,0418  & 2,0418  & 0,001\%  & 2,1698  & 5,522\%  & 2,0507  & 5,765\% \tabularnewline
2m  & 3,3126  & 3,3126  & 0,001\%  & 3,3997  & 5,308\%  & 3,3231  & 5,541\% \tabularnewline
3m  & 4,2047  & 4,2049  & 0,004\%  & 4,0221  & 5,531\%  & 4,1822  & 5,630\% \tabularnewline
6m  & 6,1892  & 6,1895  & 0,005\%  & 5,9304  & 5,685\%  & 6,2124  & 5,795\% \tabularnewline
9m  & 7,6209  & 7,6210  & 0,002\%  & 7,8537  & 5,953\%  & 7,7022  & 6,134\% \tabularnewline
1y  & 8,6815  & 8,6828  & 0,015\%  & 8,4780  & 6,330\%  & 8,8203  & 6,549\% \tabularnewline
\hline 
EPE  & 6,3807  & 6,3811  & 0,007\%  & 6,3196  & -0,957\%  & 6,4407  & 0,941\% \tabularnewline
\hline 
\end{tabular}
\par\end{centering}

{\footnotesize{}{}{\protect\protect\caption{{\footnotesize{}{}{EPE: $S_{0}=100,\sigma=30\%.$ \label{tab:100.30}}}}
}{\footnotesize \par}

{\footnotesize{}}}{\footnotesize \par}

{\footnotesize{}{}} 
\end{table}

\begin{table}[H]
\begin{centering}
{\footnotesize{}{}{}}%
\begin{tabular}{|c|c|cc|cc|cc|}
\cline{2-8} 
\multicolumn{1}{c|}{} & {\footnotesize{}{}{MC-Sobol $10^{6}$ }}  & \multicolumn{2}{c|}{{\footnotesize{}{}{Quantization}}} & \multicolumn{2}{c|}{{\footnotesize{}{}Monte Carlo}} & \multicolumn{2}{c|}{{\footnotesize{}{}{MC-Sobol $10^{3}$}}}\tabularnewline
\hline 
{\footnotesize{}{}{t }}  & {\footnotesize{}{}{EE }}  & \multicolumn{1}{c}{{\footnotesize{}{}{EE}}} & {\footnotesize{}{}{$\varepsilon$ }}  & \multicolumn{1}{c}{{\footnotesize{}{}{EE}}} & {\footnotesize{}{}{RSD }}  & \multicolumn{1}{c}{{\footnotesize{}{}{EE}}} & {\footnotesize{}{}{RSD }} \tabularnewline
\hline 
1w  & 5,9931  & 5,9931  & 0,000\%  & 5,9450  & 0,786\%  & 5,9939  & 0,777\% \tabularnewline
2w  & 5,9972  & 5,9972  & 0,000\%  & 6,0020  & 1,095\%  & 5,9974  & 1,105\% \tabularnewline
3w  & 6,0053  & 6,0053  & 0,000\%  & 6,0699  & 1,372\%  & 6,0046  & 1,350\% \tabularnewline
1m  & 6,0194  & 6,0194  & 0,000\%  & 6,1167  & 1,537\%  & 6,0213  & 1,546\% \tabularnewline
2m  & 6,1483  & 6,1482  & -0,001\%  & 6,3028  & 2,048\%  & 6,1578  & 2,136\% \tabularnewline
3m  & 6,3084  & 6,3083  & -0,001\%  & 6,2873  & 2,407\%  & 6,3192  & 2,483\% \tabularnewline
6m  & 6,7919  & 6,7918  & -0,002\%  & 6,7474  & 3,050\%  & 6,7916  & 3,188\% \tabularnewline
9m  & 7,1835  & 7,1835  & 0,000\%  & 7,2932  & 3,677\%  & 7,2228  & 3,721\% \tabularnewline
1y  & 7,5898  & 7,5915  & 0,022\%  & 7,4585  & 4,085\%  & 7,5753  & 4,210\% \tabularnewline
\hline 
EPE  & 6,9306  & 6,9310  & 0,005\%  & 6,9284  & -0,031\%  & 6,9385  & 0,114\% \tabularnewline
\hline 
\end{tabular}
\par\end{centering}

{\footnotesize{}{}{\protect\protect\caption{{\footnotesize{}{}{EPE: $S_{0}=110,\sigma=15\%.$ \label{tab:110.15}}}}
}{\footnotesize \par}

{\footnotesize{}}}{\footnotesize \par}

{\footnotesize{}{}} 
\end{table}

\begin{table}[H]
\begin{centering}
{\footnotesize{}{}{}}%
\begin{tabular}{|c|c|cc|cc|cc|}
\cline{2-8} 
\multicolumn{1}{c|}{} & {\footnotesize{}{}{MC-Sobol $10^{6}$ }}  & \multicolumn{2}{c|}{{\footnotesize{}{}{Quantization}}} & \multicolumn{2}{c|}{{\footnotesize{}{}Monte Carlo}} & \multicolumn{2}{c|}{{\footnotesize{}{}{MC-Sobol $10^{3}$}}}\tabularnewline
\hline 
{\footnotesize{}{}{t }}  & {\footnotesize{}{}{EE }}  & \multicolumn{1}{c}{{\footnotesize{}{}{EE}}} & {\footnotesize{}{}{$\varepsilon$ }}  & \multicolumn{1}{c}{{\footnotesize{}{}{EE}}} & {\footnotesize{}{}{RSD }}  & \multicolumn{1}{c}{{\footnotesize{}{}{EE}}} & {\footnotesize{}{}{RSD }} \tabularnewline
\hline 
1w  & 6,5122  & 6,5122  & 0,000\%  & 6,4141  & 1,465\%  & 6,5134  & 1,439\% \tabularnewline
2w  & 6,5989  & 6,5989  & 0,000\%  & 6,6139  & 1,926\%  & 6,6015  & 1,950\% \tabularnewline
3w  & 6,7315  & 6,7315  & 0,000\%  & 6,8607  & 2,310\%  & 6,7373  & 2,273\% \tabularnewline
1m  & 6,8813  & 6,8813  & 0,000\%  & 7,0735  & 2,479\%  & 6,8863  & 2,514\% \tabularnewline
2m  & 7,5927  & 7,5927  & 0,000\%  & 7,8513  & 2,975\%  & 7,5995  & 3,131\% \tabularnewline
3m  & 8,1894  & 8,1895  & 0,001\%  & 8,1085  & 3,360\%  & 8,2180  & 3,442\% \tabularnewline
6m  & 9,6447  & 9,6446  & -0,001\%  & 9,4633  & 3,984\%  & 9,6789  & 4,093\% \tabularnewline
9m  & 10,7365  & 10,7363  & -0,002\%  & 10,9562  & 4,531\%  & 10,7812  & 4,626\% \tabularnewline
1y  & 11,5732  & 11,5745  & 0,012\%  & 11,4042  & 4,941\%  & 11,5706  & 5,132\% \tabularnewline
\hline 
EPE  & 9,8664  & 9,8666  & 0,003\%  & 9,8547  & -0,118\%  & 9,8887  & 0,226\% \tabularnewline
\hline 
\end{tabular}
\par\end{centering}

{\footnotesize{}{}{\protect\protect\caption{{\footnotesize{}{}{EPE: $S_{0}=110,\sigma=25\%.$ \label{tab:110.25}}}}
}{\footnotesize \par}

{\footnotesize{}}}{\footnotesize \par}

{\footnotesize{}{}} 
\end{table}

\begin{table}[H]
\begin{centering}
{\footnotesize{}{}{}}%
\begin{tabular}{|c|c|cc|cc|cc|}
\cline{2-8} 
\multicolumn{1}{c|}{} & {\footnotesize{}{}{MC-Sobol $10^{6}$ }}  & \multicolumn{2}{c|}{{\footnotesize{}{}{Quantization}}} & \multicolumn{2}{c|}{{\footnotesize{}{}Monte Carlo}} & \multicolumn{2}{c|}{{\footnotesize{}{}{MC-Sobol $10^{3}$}}}\tabularnewline
\hline 
{\footnotesize{}{}{t }}  & {\footnotesize{}{}{EE }}  & \multicolumn{1}{c}{{\footnotesize{}{}{EE}}} & {\footnotesize{}{}{$\varepsilon$ }}  & \multicolumn{1}{c}{{\footnotesize{}{}{EE}}} & {\footnotesize{}{}{RSD }}  & \multicolumn{1}{c}{{\footnotesize{}{}{EE}}} & {\footnotesize{}{}{RSD }} \tabularnewline
\hline 
1w  & 6,7056  & 6,7056  & 0,000\%  & 6,5844  & 1,762\%  & 6,7063  & 1,727\% \tabularnewline
2w  & 6,8948  & 6,8948  & 0,000\%  & 6,9154  & 2,237\%  & 6,8984  & 2,270\% \tabularnewline
3w  & 7,1282  & 7,1282  & 0,000\%  & 7,3020  & 2,627\%  & 7,1311  & 2,602\% \tabularnewline
1m  & 7,3679  & 7,3680  & 0,000\%  & 7,6149  & 2,783\%  & 7,3787  & 2,836\% \tabularnewline
2m  & 8,3987  & 8,3987  & 0,000\%  & 8,6959  & 3,263\%  & 8,4147  & 3,434\% \tabularnewline
3m  & 9,2137  & 9,2140  & 0,003\%  & 9,0911  & 3,657\%  & 9,2276  & 3,754\% \tabularnewline
6m  & 11,1493  & 11,1492  & -0,001\%  & 10,8929  & 4,279\%  & 11,1882  & 4,389\% \tabularnewline
9m  & 12,5989  & 12,5984  & -0,003\%  & 12,8948  & 4,798\%  & 12,6484  & 4,921\% \tabularnewline
1y  & 13,6675  & 13,6689  & 0,010\%  & 13,4530  & 5,231\%  & 13,7136  & 5,432\% \tabularnewline
\hline 
EPE  & 11,4158  & 11,4160  & 0,002\%  & 11,3947  & -0,185\%  & 11,4523  & 0,320\% \tabularnewline
\hline 
\end{tabular}
\par\end{centering}

{\footnotesize{}{}{\protect\protect\caption{{\footnotesize{}{}{EPE: $S_{0}=110,\sigma=30\%.$ \label{tab:110.30}}}}
}{\footnotesize \par}

{\footnotesize{}}}{\footnotesize \par}

{\footnotesize{}{}} 
\end{table}

Below we present a couple of figures which clearly show how the VQ technique performs
with an excellent accuracy over the different situations. It is worth
to observe that, for some very deep out of the money situation, the
Monte Carlo simulations shows a huge relative standard error. This
occurs when the value of EE is next to zero. We remember that when
the EE is very small, no effective counterparty risk has to be faced
by the bank, hence, the magnitude of the error is not as dramatic
as it seems from a numerical perspective.

Eventually, we plotted the percent relative error $\varepsilon$ and
the percent relative standard error RSD, in order to compare, once
again, the quantization technique and the Monte Carlo method, showing
the excellent accuracy of the former approach, when compared with
the latter.

\begin{figure}[H]
\begin{centering}
\includegraphics[scale=0.5]{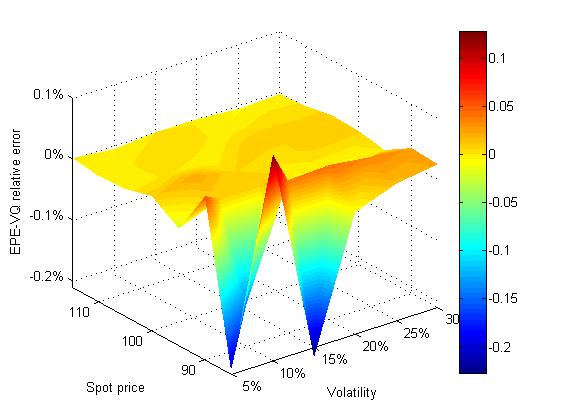} 
\par\end{centering}

{\footnotesize{}{}{{\protect\protect\caption{{\footnotesize{}{}{{VQ: EPE error $\varepsilon$ for portfolio\label{fig:EPE_VQ}}}}}
}{\footnotesize \par}

{\footnotesize{}}}{\footnotesize \par}

{\footnotesize{}{}}}{\footnotesize \par}

{\footnotesize{}{}{}} 
\end{figure}

\begin{figure}[H]
\begin{centering}
\includegraphics[scale=0.5]{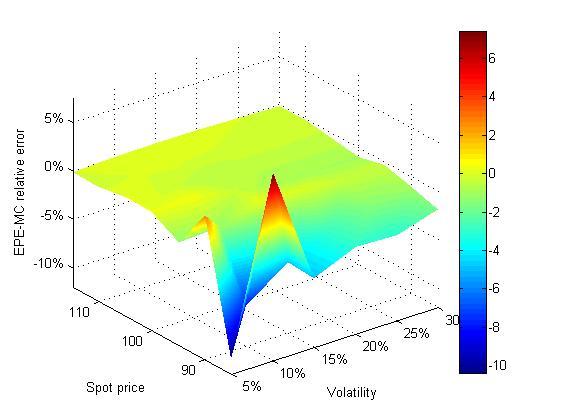} 
\par\end{centering}

{\footnotesize{}{}{{\protect\protect\caption{{\footnotesize{}{}{{MC: EPE relative error $RSD$ for portfolio\label{fig:EPE_MC}}}}}
}{\footnotesize \par}

{\footnotesize{}}}{\footnotesize \par}

{\footnotesize{}{}}}{\footnotesize \par}

{\footnotesize{}{}{}} 
\end{figure}

\section{Conclusions and Further Research}

\label{conc} In the present work we show how the quantization approach outperforms
the classical Monte Carlo methods in the CCR field. The counterparty
risk field poses a lot of theoretical and practical challenges. A
whole portfolio of derivatives must be evaluated in the future, for
several time steps and many scenarios, in order to get some useful
risk figures. This large amount of computations can be solved by improving
the technologies or the algorithmic strategies.

At the best of our knowledge, this is the first work which exploits the
quantization approach in this area. The quantization has been
intensively tested in the last decade in some pricing problems for
different complex situations: American style options, multidimensional
assets, credit derivatives. Despite this first study covers some simple
derivatives and small dimension portfolios, the quantization seems
to be very promising. Given an equivalent computational effort, it
is undoubtedly better than the standard Monte Carlo simulation and
some of its refinements such as the Sobol sequences.

Nevertheless, further research is needed. As the next step, we aiming at treating  more involved portfolios and payoffs, as in the case of path
dependent options, where the quantization tree could be a competitor
of binomial and trinomial trees. Moreover, concerning  portfolios
depending on many underlyings, there are issues related to the choice of a {\it coherent}  set of quantized paths that has to be fixed. 

Finally, numerical extensions as to be taken into account in order to pass from the accuracy comparison, given the same computational effort, to the search of 
the  quantization tradeoff, namely an estimate of its relative effort 
saving, given the same accuracy.


\begin{thebibliography}{10}
\bibitem{bcbs0} BCBS (2006) ``\emph{International Convergence of
Capital Measurement and Capital Standards}'', \emph{Basel Committee
Paper} 128.

\bibitem{bcbs1}BCBS (2011) ``Basel III: A global regulatory framework
for more resilient banks and banking systems'', \emph{ {Basel Committee
Paper}}, 189

\bibitem{bcbs2}BCBS (2014) ``The standardized approach for measuring
counterparty credit risk exposures'', \emph{Basel Committee Paper}
279

\bibitem{pag1} Bally V., Pag\`es G., Printemps J. (2010) ``A quantization
tree method for pricing and hedging multi-dimensional American options'',
\emph{Working Paper}

\bibitem{bs} Black F., Scholes M. (1973) ``The Pricing of Options
and Corporate Liabilities'', \emph{ Journal of Political Economy}
81(3), 637--654.

\bibitem{BucWise} Bucklew J.A., Wise G.L. (1982) ``Multidimensional
asymptotic quantization theory with $r^{th}$ power distortion'',
\emph{IEEE Trans. Inform. Theory}, 28(2), 239--247.

\bibitem{Caflish} Caflisch R.E., Morokoff, W.J. (1995) ``Quasi-Monte
Carlo integration'', \emph{J. Comput. Phys.} 122(2), 218--230.

\bibitem{cas}Castagna A. (2012) ``Fast computing in the CCR and
CVA measurement'', \emph{IASON ALGO}, Working Paper.

\bibitem{Ces} Cesari G. (2009) \emph{Modeling, Pricing and Hedging
Counterparty Credit Exposure}, Springer Finance Editor.

\bibitem{Dup}Dupire B. (1994), ``Pricing with a smile'', \emph{Risk},
January, 18--20.

\bibitem{ger}Gersho A., Gray R. (1982) ``IEEE on Information Theory,
Special Issue on Quantization'', 28.

\bibitem{Hag} Hagan P. et al (2002), ``Managing Smile Risk'', \emph{Wilmott
Magazine}.

\bibitem{Hes}Heston S. (1993), ``A Closed-Form Solution for Options
with Stochastic Volatility with Applications to Bond and Currency
Options'', \emph{The Review of Financial Studies} 6(2), 327--343.

\bibitem{IFRS}IFRS (2013), ``IFRS 13 Fair Value Measurement'',
\emph{ {IFRS Technical Paper}}

\bibitem{Mikes} A.Mikes, (2013), ``The Appeal of the Appropriate:
Accounting, Risk Management, and the Competition for the Supply of
Control Systems'', Harward Business School - working papers, 115(12).

\bibitem{pag2} Pag\`es G., Printemps J., Pham H. (2004), ``Optimal
quantization methods and applications to numerical problems in finance'',
\emph{Handbook on Numerical Methods in Finance}, Birkh\:{a}user,
253--298.

\bibitem{pag4} Pag\`es G., Luschgy H. (2006), ``Functional quantization
of a class of Brownian diffusions: A constructive approach'', \emph{Stochastic
processes and their Applications}, Elsevier, 116, 310--336.

\bibitem{pag3} Pag\`es G., Wilbertz B. (2012), ``Intrinsic stationarity
for vector quantization: Foundation of dual quantization'', \emph{SIAM
Journal on Numerical Analysis}, 747--780.

\bibitem{pag5} Pag\`es G., Wilbertz B. (2011), ``GPGPUs in computational
finance: Massive parallel computing for American style options'',
\emph{ {Working paper}}.

\bibitem{pyk1} Pykhtin M., Zhu S. (2007), ``A Guide to Modelling
Counterparty Credit Risk'', GARP publication.

\bibitem{Sellami} Sellami, A. (2005) ``M\`ethodes de quantification
optimale pour le filtrage et applications \'a la finance'', Applied
mathematics: Universit\`e Paris Dauphine.

\bibitem{Shreve} Shreve, S., (2004) \emph{Stochastic Calculus for
Finance II. Continuous time models}, Springer Finance series

\bibitem{Zad1} Zador, P.L. (1963), \emph{Development and evaluation
of procedures for quantizing multivariate distributions}, Ph.D. dissertation,
Stanford Univ. (USA).

\bibitem{Zad2} Zador, P.L. (1982), ``Asymptotic quantization error
of continuous signals and the quantization dimension'', \emph{IEEE
Trans. Inform. Theory}, 28(2). 
\end{thebibliography}
\end{document}